\documentclass{amsart}%
\usepackage{amsfonts}
\usepackage{amsmath}
\usepackage{amssymb}
\usepackage{graphicx}%
\newtheorem{theorem}{Theorem}
\theoremstyle{plain}

\newtheorem{definition}{Definition}

\newtheorem{lemma}{Lemma}

\newtheorem{proposition}{Proposition}
\newtheorem{remark}{Remark}

\numberwithin{equation}{section}
\ifx\pdfoutput\relax\let\pdfoutput=\undefined\fi
\newcount\msipdfoutput
\ifx\pdfoutput\undefined\else
\ifcase\pdfoutput\else
\ifx\paperwidth\undefined\else
\ifdim\paperheight=0pt\relax\else\pdfpageheight\paperheight\fi
\ifdim\paperwidth=0pt\relax\else\pdfpagewidth\paperwidth\fi
\fi\fi\fi
\begin{document}
\title[Ultrametric Diffusion and Exponential Landscapes]{Ultrametric Diffusion, Exponential Landscapes, and the First Passage Time Problem}
\author{Anselmo Torresblanca-Badillo}
\address{Departamento de Matem\'{a}ticas y Est\'{a}disticaUniversidad del Norte,
Barranquilla, Colombia}
\email{atorresblanca@uninorte.edu.co}
\author{W. A. Z\'{u}\~{n}iga-Galindo}
\address{Centro de Investigaci\'{o}n y de Estudios Avanzados del Instituto
Polit\'{e}cnico Nacional\\
Departamento de Matem\'{a}ticas, Unidad Quer\'{e}taro\\
Libramiento Norponiente \#2000, Fracc. Real de Juriquilla. Santiago de
Quer\'{e}taro, Qro. 76230\\
M\'{e}xico.}
\email{wazuniga@math.cinvestav.edu.mx}
\thanks{The second author was partially supported by Conacyt Grant No. 250845.}

\begin{abstract}
In this article we study certain ultradiffusion equations connected with
energy landscapes of exponential type. These equations are connected with the
$p$-adic models of complex systems introduced by Avetisov et al. We show that
the fundamental solutions of these equations are transition density functions
of L\'{e}vy processes with state space $\mathbb{Q}_{p}^{n}$, we also
study\ some aspects of these processes including the first passage time problem.

\end{abstract}
\maketitle

\section{Introduction}

Stochastic processes on ultrametric spaces have received a lot attention in
the latest years due to their connections with models of complex systems, see
e.g. \cite{A-K3}, \cite{Av-2}-\cite{Av-7}, \cite{Casas-Zuniga}, \cite{Ch-Z-2}%
-\cite{Ch-Z-1}, \cite{Chen-Kumagi}, \cite{Dra-Kh-K-V}, \cite{Evans},
\cite{Hoffmann}-\cite{Kozyrev SV}, \cite{R-Zu}, \cite{T-Z}-\cite{V-V-Z},
\cite{Yoshino}-\cite{Zuniga-LNM-2016}, and the references therein. A central
paradigm in physics of complex systems (such proteins or glasses) asserts that
the dynamics of such systems can be modeled as a random walk in the energy
landscape of the system, see e.g. \cite{Fraunfelder et al}-\cite{Fraunfelder
et al 3}, \cite{KKZuniga}, \cite{Kozyrev SV}, and the references
therein.\ Typically these landscapes have a huge number of local minima. It is
clear that a description of the dynamics on such landscapes require an
adequate approximation. The interbasin kinetics offers an acceptable solution
to this problem. By using this approach an energy landscape is approximated by
an ultrametric space (a rooted tree) and a function on this space describing
the distribution of the activation barriers, see e.g. \cite{Becker et al},
\cite{Stillinger et al 1}-\cite{Stillinger et al 2}. In this setting the
dynamics of a complex system is codified in a master equation which describes
the temporal behavior of the jumping probability between two states of the
system, see e.g. \cite{Kozyrev SV}. In \cite{Av-4}-\cite{Av-5} Avetisov et al.
introduced a new class of models for complex systems based on $p$-adic
analysis, these models can be applied, for instance, to study the relaxation
of biological complex systems. In this article we continue the study of these
models, more precisely, we study $n$-dimensional versions the master equations
introduced in \cite{Av-4}-\cite{Av-2}, for exponential landscapes. We
establish rigorously that such equations are ultradiffusion equations, i.e. we
show that the fundamental solutions of these equations are transition density
functions of L\'{e}vy processes with space state $\mathbb{Q}_{p}^{n}$, see
Theorem \ref{Theorem2}. We also study the first passage time problem for the
processes constructed in this article, see Theorem \ref{Theorem3}. Finally, we
point out that it is not possible, due to physical and mathematical reasons,
to forget the ultradiffusion equations and work exclusively with the attached
Markovian semigroups. These semigroups have been extensively studied in the
case of totally disconnected groups, see e.g. \cite{Bendikov}-\cite{Bendikov2}%
, \cite{Evans}. For instance, in the study of the first passage time problem,
the ultradiffusion equation itself plays a central role, see Lemmas
\ref{lemma8}-\ref{lemma9}. In a long term perspective, the goal is to extend
the results presented here to equations of variable coefficients to obtain a
theory similar to the one presented in \cite{Zuniga-LNM-2016}, \cite{Ch-Z-2}.
A detailed discussion of our results (from the perspective of the Avetisov et
al. models) and as well as a comparison with current literature is given in
Section 3.

\section{$p$\textbf{-}Adic Analysis: Essential Ideas}

\subsection{The field of $p$-adic numbers}

Along this article $p$ will denote a prime number. The field of $p-$adic
numbers $\mathbb{Q}_{p}$ is defined as the completion of the field of rational
numbers $\mathbb{Q}$ with respect to the $p-$adic norm $|\cdot|_{p}$, which is
defined as
\[
\left\vert x\right\vert _{p}=\left\{
\begin{array}
[c]{lll}%
0 & \text{if} & x=0\\
&  & \\
p^{-\gamma} & \text{if} & x=p^{\gamma}\frac{a}{b}\text{,}%
\end{array}
\right.
\]
where $a$ and $b$ are integers coprime with $p$. The integer $\gamma:=ord(x)$,
with $ord(0):=+\infty$, is called the\textit{\ }$p-$\textit{adic order of}
$x$. We extend the $p-$adic norm to $\mathbb{Q}_{p}^{n}$ by taking%
\[
||x||_{p}:=\max_{1\leq i\leq n}|x_{i}|_{p},\qquad\text{for }x=(x_{1}%
,\dots,x_{n})\in\mathbb{Q}_{p}^{n}.
\]
We define $ord(x)=\min_{1\leq i\leq n}\{ord(x_{i})\}$, then $||x||_{p}%
=p^{-ord(x)}$.\ The metric space $\left(  \mathbb{Q}_{p}^{n},||\cdot
||_{p}\right)  $ is a complete ultrametric space. As a topological space
$\mathbb{Q}_{p}$\ is homeomorphic to a Cantor-like subset of the real line,
see e.g. \cite{Alberio et al}, \cite{V-V-Z}.

Any $p-$adic number $x\neq0$ has a unique expansion of the form
\[
x=p^{ord(x)}\sum_{j=0}^{\infty}x_{i}p^{j},
\]
where $x_{j}\in\{0,1,2,\dots,p-1\}$ and $x_{0}\neq0$. By using this expansion,
we define \textit{the fractional part of }$x\in\mathbb{Q}_{p}$, denoted
$\{x\}_{p}$, as the rational number
\[
\left\{  x\right\}  _{p}=\left\{
\begin{array}
[c]{lll}%
0 & \text{if} & x=0\text{ or }ord(x)\geq0\\
&  & \\
p^{ord(x)}\sum_{j=o}^{-ord_{p}(x)-1}x_{j}p^{j} & \text{if} & ord(x)<0.
\end{array}
\right.
\]

In addition, any $p-$adic number $x\neq0$ can be represented uniquely as
$x=p^{ord(x)}ac\left(  x\right)  $ where $ac\left(  x\right)  =\sum
_{j=0}^{\infty}x_{i}p^{j}$, $x_{0}\neq0$, is called the \textit{angular
component} of $x$. Notice that $\left\vert ac\left(  x\right)  \right\vert
_{p}=1$.

\subsection{Additive characters}

Set $\chi_{p}(y)=\exp(2\pi i\{y\}_{p})$ for $y\in\mathbb{Q}_{p}$. The map
$\chi_{p}(\cdot)$ is an additive character on $\mathbb{Q}_{p}$, i.e. a
continuous map from $\left(  \mathbb{Q}_{p},+\right)  $ into $S$ (the unit
circle considered as multiplicative group) satisfying $\chi_{p}(x_{0}%
+x_{1})=\chi_{p}(x_{0})\chi_{p}(x_{1})$, $x_{0},x_{1}\in\mathbb{Q}_{p}$. The
additive characters of $\mathbb{Q}_{p}$ form an Abelian group which is
isomorphic to $\left(  \mathbb{Q}_{p},+\right)  $, the isomorphism is given by
$\xi\rightarrow\chi_{p}(\xi x)$, see e.g. \cite[Section 2.3]{Alberio et al}.

\subsection{Topology of $\mathbb{Q}_{p}^{n}$}

For $r\in\mathbb{Z}$, denote by $B_{r}^{n}(a)=\{x\in\mathbb{Q}_{p}%
^{n};||x-a||_{p}\leq p^{r}\}$ \textit{the ball of radius }$p^{r}$ \textit{with
center at} $a=(a_{1},\dots,a_{n})\in\mathbb{Q}_{p}^{n}$, and take $B_{r}%
^{n}(0):=B_{r}^{n}$. Note that $B_{r}^{n}(a)=B_{r}(a_{1})\times\cdots\times
B_{r}(a_{n})$, where $B_{r}(a_{i}):=\{x\in\mathbb{Q}_{p};|x_{i}-a_{i}|_{p}\leq
p^{r}\}$ is the one-dimensional ball of radius $p^{r}$ with center at
$a_{i}\in\mathbb{Q}_{p}$. The ball $B_{0}^{n}$ equals the product of $n$
copies of $B_{0}=\mathbb{Z}_{p}$, \textit{the ring of }$p- $\textit{adic
integers}. We also denote by $S_{r}^{n}(a)=\{x\in\mathbb{Q}_{p}^{n}%
;||x-a||_{p}=p^{r}\}$ \textit{the sphere of radius }$p^{r}$ \textit{with
center at} $a=(a_{1},\dots,a_{n})\in\mathbb{Q}_{p}^{n}$, and take $S_{r}%
^{n}(0):=S_{r}^{n}$. We notice that $S_{0}^{1}=\mathbb{Z}_{p}^{\times}$ (the
group of units of $\mathbb{Z}_{p}$), but $\left(  \mathbb{Z}_{p}^{\times
}\right)  ^{n}\subsetneq S_{0}^{n}$. The balls and spheres are both open and
closed subsets in $\mathbb{Q}_{p}^{n}$. In addition, two balls in
$\mathbb{Q}_{p}^{n}$ are either disjoint or one is contained in the other.

As a topological space $\left(  \mathbb{Q}_{p}^{n},||\cdot||_{p}\right)  $ is
totally disconnected, i.e. the only connected \ subsets of $\mathbb{Q}_{p}%
^{n}$ are the empty set and the points. A subset of $\mathbb{Q}_{p}^{n}$ is
compact if and only if it is closed and bounded in $\mathbb{Q}_{p}^{n}$, see
e.g. \cite[Section 1.3]{V-V-Z}, or \cite[Section 1.8]{Alberio et al}. The
balls and spheres are compact subsets. Thus $\left(  \mathbb{Q}_{p}%
^{n},||\cdot||_{p}\right)  $ is a locally compact topological space.

We will use $\Omega\left(  p^{-r}||x-a||_{p}\right)  $ to denote the
characteristic function of the ball $B_{r}^{n}(a)$. We will use the notation
$1_{A}$ for the characteristic function of a set $A$.

\subsection{The Bruhat-Schwartz space and the Fourier transform}

A complex-valued function $\varphi$ defined on $\mathbb{Q}_{p}^{n}$ is
\textit{called locally constant} if for any $x\in\mathbb{Q}_{p}^{n}$ there
exist an integer $l(x)\in\mathbb{Z}$ such that%
\begin{equation}
\varphi(x+x^{\prime})=\varphi(x)\text{ for }x^{\prime}\in B_{l(x)}%
^{n}.\label{local_constancy}%
\end{equation}

A function $\varphi:\mathbb{Q}_{p}^{n}\rightarrow\mathbb{C}$ is called a
\textit{Bruhat-Schwartz function (or a test function)} if it is locally
constant with compact support. Any test function can be represented as a
linear combination, with complex coefficients, of characteristic functions of
balls. The $\mathbb{C}$-vector space of Bruhat-Schwartz functions is denoted
by $\mathcal{D}(\mathbb{Q}_{p}^{n})$. For $\varphi\in\mathcal{D}%
(\mathbb{Q}_{p}^{n})$, the largest number $l=l(\varphi)$ satisfying
(\ref{local_constancy}) is called \textit{the exponent of local constancy (or
the parameter of constancy) of} $\varphi$.

If $U$ is an open subset of $\mathbb{Q}_{p}^{n}$, $\mathcal{D}(U)$ denotes the
space of test functions with supports contained in $U$, then $\mathcal{D}(U)$
is dense in
\[
L^{\rho}\left(  U\right)  =\left\{  \varphi:U\rightarrow\mathbb{C};\left\Vert
\varphi\right\Vert _{\rho}=\left\{  \int_{\mathbb{Q}_{p}^{n}}\left\vert
\varphi\left(  x\right)  \right\vert ^{\rho}d^{n}x\right\}  ^{\frac{1}{\rho}%
}<\infty\right\}  ,
\]
where $d^{n}x$ is the Haar measure on $\mathbb{Q}_{p}^{n}$ normalized by the
condition $vol(B_{0}^{n})\allowbreak=1$, for $1\leq\rho<\infty$, see e.g.
\cite[Section 4.3]{Alberio et al}.

\subsection{The Fourier transform of test functions}

Given $\xi=(\xi_{1},\dots,\xi_{n})$ and $y=(x_{1},\dots,x_{n})\in
\mathbb{Q}_{p}^{n}$, we set $\xi\cdot x:=\sum_{j=1}^{n}\xi_{j}x_{j}$. The
Fourier transform of $\varphi\in\mathcal{D}(\mathbb{Q}_{p}^{n})$ is defined
as
\[
(\mathcal{F}\varphi)(\xi)=\int_{\mathbb{Q}_{p}^{n}}\chi_{p}(\xi\cdot
x)\varphi(x)d^{n}x\quad\text{for }\xi\in\mathbb{Q}_{p}^{n},
\]
where $d^{n}x$ is the normalized Haar measure on $\mathbb{Q}_{p}^{n}$. The
Fourier transform is a linear isomorphism from $\mathcal{D}(\mathbb{Q}_{p}%
^{n})$ onto itself satisfying $(\mathcal{F}(\mathcal{F}\varphi))(\xi
)=\varphi(-\xi)$, see e.g. \cite[Section 4.8]{Alberio et al}. We will also use
the notation $\mathcal{F}_{x\rightarrow\xi}\varphi$ and $\widehat{\varphi}%
$\ for the Fourier transform of $\varphi$.

\subsection{Distributions}

Let $\mathcal{D}^{\prime}(\mathbb{Q}_{p}^{n})$ denote the $\mathbb{C}$-vector
space of all continuous functionals (distributions) on $\mathcal{D}%
(\mathbb{Q}_{p}^{n})$. The natural pairing $\mathcal{D}^{\prime}%
(\mathbb{Q}_{p}^{n})\times\mathcal{D}(\mathbb{Q}_{p}^{n})\rightarrow
\mathbb{C}$ is denoted as $\left(  T,\varphi\right)  $ for $T\in
\mathcal{D}^{\prime}(\mathbb{Q}_{p}^{n})$ and $\varphi\in\mathcal{D}%
(\mathbb{Q}_{p}^{n})$, see e.g. \cite[Section 4.4]{Alberio et al}.

Every $f$\ in $L_{loc}^{1}$ defines a distribution $f\in\mathcal{D}^{\prime
}\left(  \mathbb{Q}_{p}^{n}\right)  $ by the formula
\[
\left(  f,\varphi\right)  =%
{\textstyle\int\limits_{\mathbb{Q}_{p}^{n}}}
f\left(  x\right)  \varphi\left(  x\right)  d^{n}x.
\]
Such distributions are called \textit{regular distributions}.

\subsection{The Fourier transform of a distribution}

The Fourier transform $\mathcal{F}\left[  T\right]  $ of a distribution
$T\in\mathcal{D}^{\prime}\left(  \mathbb{Q}_{p}^{n}\right)  $ is defined by%
\[
\left(  \mathcal{F}\left[  T\right]  ,\varphi\right)  =\left(  T,\mathcal{F}%
\left[  \varphi\right]  \right)  \text{ for all }\varphi\in\mathcal{D}%
(\mathbb{Q}_{p}^{n})\text{.}%
\]
The Fourier transform $T\rightarrow\mathcal{F}\left[  T\right]  $ is a linear
(and continuous) isomorphism from $\mathcal{D}^{\prime}\left(  \mathbb{Q}%
_{p}^{n}\right)  $\ onto $\mathcal{D}^{\prime}\left(  \mathbb{Q}_{p}%
^{n}\right)  $. Furthermore, $T=\mathcal{F}\left[  \mathcal{F}\left[
T\right]  \left(  -\xi\right)  \right]  $.

\section{\label{Sect_Disc}$p$-adic Models of Relaxation of Complex Systems:
discussion of the results.}

The dynamics of a large class of complex systems such as glasses and proteins
is described by a random walk on a complex energy landscape, see e.g.
\cite{Fraunfelder et al 2}, \cite{Fraunfelder et al 3}, \cite[and \ the
references therein]{KKZuniga}, \cite{Kozyrev SV}, \cite{Wales}.\ An
\textit{energy landscape} (or simply a landscape) is a continuous function
$\mathbb{U}:X\rightarrow\mathbb{R}$ that assigns to each physical state of a
system its energy. In many cases we can take $X$ to be a subset of
$\mathbb{R}^{N}$. The term \textit{complex landscape} means that function
$\mathbb{U}$ has many local minima. In this case the method of
\textit{interbasin kinetics} is applied, in this approach, the study of a
random walk on a complex landscape is based on a description of the kinetics
generated by transitions between groups of states (\textit{basins}). Minimal
basins correspond to local minima of energy, and large basins have
hierarchical structure. The transition rate between basins is determined by
the Arrhenius factor, which depends on the energy barrier between these
basins. Procedures for constructing hierarchies of basins kinetics from any
energy landscapes have been studied extensively, see e.g. \cite{Becker et al},
\cite{Stillinger et al 1}, \cite{Stillinger et al 2}. By using these methods,
a complex landscape is approximated by a \textit{disconnectivity graph} (a
rooted tree) and by a function on the tree describing the distribution of the
activation energies. The dynamics of the system is then encoded in a system of
kinetic equations of the form:%
\begin{equation}
\frac{\partial}{\partial t}u\left(  i,t\right)  =-\sum_{j}\left\{
T(i,j)u\left(  i,t\right)  -T\left(  j,i\right)  u\left(  j,t\right)
\right\}  v\left(  j\right)  ,\label{Equation_Discrete}%
\end{equation}
where the indices $i$, $j$ number the states of the system (which correspond
to local minima of energy), $T(i,j)\geq0$ is the probability per unit time of
a transition from $i$ to $j$, and the $v(j)>0$ are the basin volumes. For
further details the reader may consult \cite{KKZuniga}, \cite{Kozyrev SV}, and
\ the references therein. Several models of interbasin kinetics and
hierarchical dynamics have been studied, see e.g. \cite{Hoffmann}, \cite[and
\ the references therein]{Kozyrev SV}, \cite{Ogielski et al}.

In \cite{Av-4}-\cite{Av-5} Avetisov et al. developed new class of models of
interbasin kinetics using ultrametric diffusion generated by $p$-adic
pseudodifferential operators. In these models, the time-evolution of the
system is controlled by a master equation of the form:
\begin{equation}
\frac{\partial u\left(  x,t\right)  }{\partial t}=%
{\displaystyle\int\limits_{\mathbb{Q}_{p}}}
\left\{  j\left(  x\mid y\right)  u\left(  y,t\right)  -j\left(  y\mid
x\right)  u\left(  x,t\right)  \right\}  dy\text{, }x\in\mathbb{Q}_{p}\text{,
}t\in\mathbb{R}_{+},\label{Master_E}%
\end{equation}
where the function $u\left(  x,t\right)  :\mathbb{Q}_{p}\times\mathbb{R}%
_{+}\rightarrow\mathbb{R}_{+}$ is a probability density distribution, and the
function $j\left(  x\mid y\right)  :\mathbb{Q}_{p}\times\mathbb{Q}%
_{p}\rightarrow\mathbb{R}_{+}$ is the probability of transition from state $y$
to the state $x$ per unit time. Master equation (\ref{Master_E}) is a
continuous version of (\ref{Equation_Discrete}) obtained from it by passing to
a `continuous limit' in $\mathbb{Q}_{p}$ under the conditions $v(j)=1$,
$T(i,j)=j\left(  \left\vert i-j\right\vert _{p}\right)  $, see e.g.
\cite{Kozyrev SV}. The transition from a state $y$ to a state $x$ can be
perceived as overcoming the energy barrier \ separating these states. In
\cite{Av-4} an Arrhenius type relation was used, that is,%
\[
j\left(  x\mid y\right)  \sim A(T)\exp\left\{  -\frac{\mathbb{U}\left(  x\mid
y\right)  }{kT}\right\}  ,
\]
where $\mathbb{U}\left(  x\mid y\right)  $ is the height of the activation
barrier for the transition from the state $y$ to state $x$, $k$ is the
Boltzmann constant and $T$ is the temperature. This formula establishes a
relation between the structure of the energy landscape $\mathbb{U}\left(
x\mid y\right)  $ and the transition function $j\left(  x\mid y\right)  $. The
case $j\left(  x\mid y\right)  =j\left(  y\mid x\right)  $ corresponds to a
\textit{degenerate energy landscape}. In this case the master equation
(\ref{Master_E}) takes the form%
\begin{equation}
\frac{\partial u\left(  x,t\right)  }{\partial t}=%
{\displaystyle\int\limits_{\mathbb{Q}_{p}}}
j\left(  \left\vert x-y\right\vert _{p}\right)  \left\{  u\left(  y,t\right)
-u\left(  x,t\right)  \right\}  dy\text{,}\label{Master_E1}%
\end{equation}
where $j\left(  \left\vert x-y\right\vert _{p}\right)  =\frac{A(T)}{\left\vert
x-y\right\vert _{p}}\exp\left\{  -\frac{\mathbb{U}\left(  \left\vert
x-y\right\vert _{p}\right)  }{kT}\right\}  $. By choosing $\mathbb{U}$
conveniently, several energy landscapes can be obtained. Following
\cite{Av-4}, there are three basic landscapes: (i) (logarithmic) $j\left(
\left\vert x-y\right\vert _{p}\right)  =\frac{1}{\left\vert x-y\right\vert
_{p}\ln^{\alpha}\left(  1+\left\vert x-y\right\vert _{p}\right)  }$,
$\alpha>1$ (ii) (linear) $j\left(  \left\vert x-y\right\vert _{p}\right)
=\frac{1}{\left\vert x-y\right\vert _{p}^{\alpha+1}}$, $\alpha>0$, (iii)
(exponential) $j\left(  \left\vert x-y\right\vert _{p}\right)  =\frac
{e^{-\alpha\left\vert x-y\right\vert _{p}}}{\left\vert x-y\right\vert _{p}}$,
$\alpha>0$.

In our opinion, the novelty and relevance of the idealistic models of Avetisov
et al. come from two facts: first, they codify, in a mathematical language,
the central physical paradigm asserting that the dynamics of several complex
systems can be described as a random walk on a rooted tree; second, these
models give a description of the characteristic types of relaxation of complex systems.

The original models of Avetisov et al. were formulated in dimension one. The
corresponding master equations were obtained by studying a random walk on a
disconnectivity graph, which comes from `one' cross section of the energy
landscape of a system, by using the above mentioned limit process, the tree
becomes in $\mathbb{Q}_{p}$. In \cite{Fraunfelder et al} Frauenfelder et al.
have explicitly pointed out that using rooted trees (disconnectivity graphs)
constructed from `one' cross section of an energy landscape of a complex
system is misleading in two respects: \textquotedblleft it appears that a
transition from an initial state $i$ to a final state $j$ must follow a unique
pathway, and second entropy \ does not play a role,\textquotedblright\ see
\cite[p. 98 and figures 11.3 and 11.4]{Fraunfelder et al}. If we consider
several cross sections of an energy landscape, each of them gives rise to a
disconnectivity graph, and hence the indices $i$, $j$ in equation
(\ref{Master_E}) are vectors running on the Cartesian product of the
disconnectivity graphs. In the continuous model, see (\ref{Master_E1}), the
variables $x$, $y$ run through $\mathbb{Q}_{p}^{n}$. Therefore, the master
equations for the Avetisov et al. models should be studied in arbitrary
dimension due to physical reasons and for mathematical generality. This
program is being developed by the second author and his collaborators in
recent years, see e.g. \cite{Casas-Zuniga}, \cite{Ch-Z-2}, \cite{Ch-Z-1},
\cite{KKZuniga}, \cite{R-Zu}, \cite{T-Z}, \cite{Zu}, \cite{Z1},
\cite{Zuniga-LNM-2016}.

This article aims to study some aspects of the dynamics of `random walks'
associated with the exponential landscapes introduced in \cite{Av-4}. The
corresponding master equation takes the form%
\begin{equation}
\left\{
\begin{array}
[c]{llll}%
\frac{\partial u\left(  x,t\right)  }{\partial t}=\left(  J\ast u(\cdot
,t)\right)  \left(  x\right)  -u\left(  x,t\right)  , & x\in\mathbb{Q}_{p}%
^{n}, & t\geq0 & (A)\\
&  &  & \\
u\left(  x,0\right)  =u_{0}\left(  x\right)  , &  &  & (B)
\end{array}
\right.  \label{Cauchy}%
\end{equation}
where $J:\mathbb{Q}_{p}^{n}\rightarrow\mathbb{R}_{+}$ belongs to a family of
functions, depending on several parameters, which codify the structure of an
energy landscape. The family of landscapes studied here is an `integrable'
generalization of the one-dimensional exponential landscapes introduced in
\cite{Av-4}. The first step is to establish that (\ref{Cauchy})-(A) is an
ultradiffusion equation, i.e. that its fundamental solution is the transition
density of a Markov process with space state $\mathbb{Q}_{p}^{n}$. It is
important to mention here, that this fact has been established rigorously only
for the linear landscapes, in this case (\ref{Cauchy})-(A) becomes a $p$-adic
heat equation, and these equations and their Markov processes have been
studied intensively lately, see \cite{Casas-Zuniga}, \cite{Ch-Z-1},
\cite{Ch-Z-2}, \cite{Kigami}, \cite{Koch}, \cite{V-V-Z}, \cite{Z1}, \cite{Zu},
\cite{Zuniga-LNM-2016}, and the references therein. For the exponential and
logarithmic landscapes of \cite{Av-4}, the function $j$ is only locally
integrable, and thus the natural domain of the operator in the \ right
hand-side of (\ref{Master_E1}) is not evident, and it is not given in
\cite{Av-4}, hence the fact that (\ref{Master_E1}) is an ultradiffusion
equation in the case of exponential and logarithm landscapes was not
established in \cite{Av-4}. By imposing to the function $J$ the condition of
being integrable, the operator in the\ right hand-side of (\ref{Cauchy})-(A)
becomes a linear bounded operator on $L^{\rho}$, $1\leq\rho\leq\infty$, and in
the case of exponential-type landscapes, we show that (\ref{Cauchy})-(A) is an
ultradiffusion equation. In this article we show that the fundamental solution
of (\ref{Cauchy})-(A) is the transition density of a L\'{e}vy process, see
Theorem \ref{Theorem2}. It is worth to mention that the real counterpart of
equation (\ref{Cauchy})-(A) has been studied intensively, \ in this setting
the equation has been used to model diffusion processes, see e.g.
\cite{Andreu-Vaillo et al}. In the preface of \cite{Andreu-Vaillo et al} the
authors show that for certain Markov processes their density transition
functions satisfy (\ref{Cauchy})-(A). In this article, we investigate the
converse of this situation, in a $p$-adic setting.

We also study the first passage time problem for stochastic processes
$\mathfrak{J}\left(  t,\omega\right)  $ whose transition density functions
satisfy (\ref{Cauchy})-(A)-(B), with $u_{0}(x)$ equals to the characteristic
function of $\mathbb{Z}_{p}^{n}$. More precisely, we study the random variable
$\tau_{\mathbb{Z}_{p}^{n}}\left(  \omega\right)  $ defined as the smallest
time in which a path of $\mathfrak{J}\left(  t,\omega\right)  $ returns to
$\mathbb{Z}_{p}^{n}$. We show that every path of any of these processes is
sure to return to $\mathbb{Z}_{p}^{n},$ see Theorem \ref{Theorem3}.

It is important to mention that our results do not cover all the exponential
landscapes introduced here, this will require the study of equations of type
(\ref{Cauchy})-(A)-(B) in a more general setting. Finally the results in
\cite{Ch-Z-1} present an $n$-dimensional generalization of the linear energy
landscapes of \cite{Av-4}, this generalization was achieved by generalizing
the $p$-adic heat equations. There are several important differences between
this article and \cite{Ch-Z-1}. First, the operators (Laplacians) considered
in \cite{Ch-Z-1} are unbounded operators densely defined in suitable subspaces
the $L^{2}\left(  \mathbb{Q}_{p}^{n}\right)  $ while \ the operators studied
here are bounded operators defined in $L^{\rho}\left(  \mathbb{Q}_{p}%
^{n}\right)  $, $1\leq\rho\leq\infty$; second, the fundamental solutions in
\cite{Ch-Z-1} are integrable functions of the position while here the
fundamental solutions are distributions which makes the study of the
corresponding stochastic processes more involved.

\section{Preliminary results}

\subsection{\label{Section_Exp_Lan}Exponential Landscapes}

In this section we give several technical results for the functions $J$ that
codify the structure of the energy landscapes studied in this article.

Set $\mathbb{R}_{+}:=\{x\in\mathbb{R};x\geq0\}.$ We fix a continuous function
$J:$\textbf{\ }$\mathbb{R}_{+}\rightarrow\mathbb{R}_{+}$, and take
$J(x)=J(||x||_{p})$ for $x\in%
\mathbb{Q}
_{p}^{n}$, then $J(x)$ is a \textit{radial function} on $%
\mathbb{Q}
_{p}^{n}$. In addition, we assume that $\ \int_{%
\mathbb{Q}
_{p}^{n}}J(||x||_{p})d^{n}x=1$.

\begin{definition}
We say that function $J(\left\Vert x\right\Vert _{p})$ is of exponential type
if there exist positive real constants $A$, $B$, $C_{1}$, and a real constant
$\gamma>-n$ such that
\[
A\left\Vert x\right\Vert _{p}^{\gamma}e^{-C_{1}\left\Vert x\right\Vert _{p}%
}\leq J(\left\Vert x\right\Vert _{p})\leq B\left\Vert x\right\Vert
_{p}^{\gamma}e^{-C_{1}\left\Vert x\right\Vert _{p}}\text{, for any }%
x\in\mathbb{Q}_{p}^{n}.
\]

\end{definition}

\begin{remark}
\label{Nota-1}The condition $\gamma>-n$ is completely necessary to assure that
$J\in L^{1}$. We notice that in dimension one the function $\frac
{e^{-C\left\vert x\right\vert _{p}}}{\left\vert x\right\vert _{p}}$, $C>0$,
which was used in \cite{Av-4} as a $j$ function, is not integrable. Indeed,
assume that $\frac{e^{-C\left\vert x\right\vert _{p}}}{\left\vert x\right\vert
_{p}}\in L^{1}$, then the following integral exists:%
\begin{align}
\int\nolimits_{\mathbb{Z}_{p}}\frac{e^{-C\left\vert x\right\vert _{p}}%
}{\left\vert x\right\vert _{p}}dx  & =\int\nolimits_{p\mathbb{Z}_{p}}%
\frac{e^{-C\left\vert x\right\vert _{p}}}{\left\vert x\right\vert _{p}}%
dx+\int\nolimits_{\mathbb{Z}_{p}^{\times}}\frac{e^{-C\left\vert x\right\vert
_{p}}}{\left\vert x\right\vert _{p}}dx\nonumber\\
& =\int\nolimits_{\mathbb{Z}_{p}}\frac{e^{-Cp^{-1}\left\vert x\right\vert
_{p}}}{\left\vert x\right\vert _{p}}dx+e^{-C}\left(  1-p^{-1}\right)
.\label{Integral}%
\end{align}
Now, since $C\left\vert x\right\vert _{p}\geq Cp^{-1}\left\vert x\right\vert
_{p}$, we have $\int\nolimits_{\mathbb{Z}_{p}}\frac{e^{-C\left\vert
x\right\vert _{p}}}{\left\vert x\right\vert _{p}}dx-\int\nolimits_{\mathbb{Z}%
_{p}}\frac{e^{-Cp^{-1}\left\vert x\right\vert _{p}}}{\left\vert x\right\vert
_{p}}dx\leq0$, which contradicts (\ref{Integral}). This situation causes
several mathematical problems that we will discuss later on.
\end{remark}

\begin{lemma}
\label{lemma1}With the above notation, the following assertions hold:

\noindent(i) $\widehat{J}(\xi)$ is a real-valued, radial (i.e. $\widehat
{J}(\xi)=\widehat{J}(\left\Vert \xi\right\Vert _{p})$), and continuous
function, satisfying $|\widehat{J}(\left\Vert \xi\right\Vert _{p})|\leq1$ and
\textbf{\ }$\widehat{J}(0)=1$;

\noindent(ii) for $\xi\in\mathbb{Q}_{p}^{n}\smallsetminus\left\{  0\right\}
$,%
\[
1-\widehat{J}(\left\Vert \xi\right\Vert _{p})=\left\Vert \xi\right\Vert
_{p}^{-n}J(p\left\Vert \xi\right\Vert _{p}^{-1})+p^{n}\left\Vert
\xi\right\Vert _{p}^{-n}\sum\nolimits_{l=0}^{\infty}p^{nl}J(p^{1+l}\left\Vert
\xi\right\Vert _{p}^{-1});
\]

\noindent(iii) if $-n<\gamma<0$, then
\[
1-\widehat{J}(\left\Vert \xi\right\Vert _{p})\leq B_{1}\left\Vert
\xi\right\Vert _{p}^{-n-\gamma}e^{-C_{1}p\left\Vert \xi\right\Vert _{p}^{-1}%
}+B_{2}\left\Vert \xi\right\Vert _{p}^{-\gamma}e^{-C_{1}p\left\Vert
\xi\right\Vert _{p}^{-1}},
\]
for $\xi\in\mathbb{Q}_{p}^{n}\smallsetminus\left\{  0\right\}  $, where
$B_{1}$, $B_{2}$ are positive constants.
\end{lemma}

\begin{proof}
(i) The Fourier transform of an integrable radial function is a real-valued
continuous radial function. Notice that $|\widehat{J}\left(  \left\Vert
\xi\right\Vert _{p}\right)  |\leq\int_{\mathbb{%
\mathbb{Q}
}_{p}^{n}}|\chi_{p}\left(  \xi\cdot x\right)  |J(x)d^{n}x=\int_{\mathbb{%
\mathbb{Q}
}_{p}^{n}}J(\left\Vert x\right\Vert _{p})d^{n}x\allowbreak=1$. Now, since $J $
is integrable $\widehat{J}\left(  0\right)  =$ $\int_{\mathbb{%
\mathbb{Q}
}_{p}^{n}}J(\left\Vert x\right\Vert _{p})d^{n}x=1$.

(ii) Take $\xi=p^{ord\left(  \xi\right)  }\xi_{0}$, with $\xi_{0}=\left(
\xi_{0}^{\left(  1\right)  },\ldots,\xi_{0}^{\left(  n\right)  }\right)  $,
$\left\Vert \xi_{0}\right\Vert _{p}=1$, and $ord\left(  \xi\right)
\in\mathbb{Z}$, then
\[
\widehat{J}\left(  \left\Vert \xi\right\Vert _{p}\right)  -1=\int
_{\mathbb{Q}_{p}^{n}}J(||x||_{p})\left\{  \chi_{p}\left(  p^{ord\left(
\xi\right)  }\xi_{0}\cdot x\right)  -1\right\}  d^{n}x.
\]
By changing variables as $y_{i}=p^{ord\left(  \xi\right)  }\xi_{0}^{\left(
i\right)  }x_{i}$ for $i=1,\ldots,n$, with $d^{n}x=p^{ord\left(  \xi\right)
n}d^{n}y$, we have%
\begin{align*}
\widehat{J}\left(  \left\Vert \xi\right\Vert _{p}\right)  -1 &  =p^{ord\left(
\xi\right)  n}\int\limits_{\mathbb{Q}_{p}^{n}}J(p^{ord\left(  \xi\right)
}||y||_{p})\left\{  \chi_{p}\left(  \sum_{i=1}^{n}y_{i}\right)  -1\right\}
d^{n}y\\
&  =p^{ord\left(  \xi\right)  n}\int\limits_{\mathbb{Q}_{p}^{n}\smallsetminus
\mathbb{Z}_{p}^{n}}J(p^{ord\left(  \xi\right)  }||y||_{p})\left\{  \chi
_{p}\left(  \sum_{i=1}^{n}y_{i}\right)  -1\right\}  d^{n}y\\
&  =p^{ord\left(  \xi\right)  n}%
{\displaystyle\sum\limits_{j=1}^{\infty}}
J(p^{j+ord\left(  \xi\right)  })\int\limits_{\left\Vert y\right\Vert
_{p}=p^{j}}\left\{  \chi_{p}\left(  \sum_{i=1}^{n}y_{i}\right)  -1\right\}
d^{n}y.
\end{align*}
By changing variables as $y=p^{-j}z$, $d^{n}y=p^{nj}d^{n}z$, we get%
\begin{gather*}
\widehat{J}\left(  \left\Vert \xi\right\Vert _{p}\right)  -1=p^{ord\left(
\xi\right)  n}%
{\displaystyle\sum\limits_{j=1}^{\infty}}
p^{nj}J(p^{j+ord\left(  \xi\right)  })\int\limits_{S_{0}^{n}}\left\{  \chi
_{p}\left(  p^{-j}\sum_{i=1}^{n}z_{i}\right)  -1\right\}  d^{n}z\\
=p^{ord\left(  \xi\right)  n}%
{\displaystyle\sum\limits_{j=1}^{\infty}}
p^{nj}J(p^{j+ord\left(  \xi\right)  })\left\{
\begin{array}
[c]{lll}%
-p^{-n} & \text{if} & j\leq0\\
-1-p^{-n} & \text{if} & j=1\\
-1 & \text{if} & j\geq2
\end{array}
\right.  \\
=-\left(  1+p^{-n}\right)  p^{ord\left(  \xi\right)  n+n}J(p^{1+ord\left(
\xi\right)  })-p^{ord\left(  \xi\right)  n}%
{\displaystyle\sum\limits_{j=2}^{\infty}}
p^{nj}J(p^{j+ord\left(  \xi\right)  })\\
=-p^{ord\left(  \xi\right)  n}J(p^{1+ord\left(  \xi\right)  })-p^{ord\left(
\xi\right)  n+n}%
{\displaystyle\sum\limits_{l=0}^{\infty}}
p^{nl}J(p^{1+l+ord\left(  \xi\right)  }).
\end{gather*}

(iii) From (ii) and the fact that $J$ is of exponential type, we get that%
\begin{equation}
1-\widehat{J}(\left\Vert \xi\right\Vert _{p})\leq Bp^{\gamma}\left\Vert
\xi\right\Vert _{p}^{-n-\gamma}e^{-C_{1}p\left\Vert \xi\right\Vert _{p}^{-1}%
}+Bp^{\gamma+n}\left\Vert \xi\right\Vert _{p}^{-n-\gamma}\sum_{l=0}^{\infty
}p^{\left(  n+\gamma\right)  l}e^{-C_{1}p^{1+l}\left\Vert \xi\right\Vert
_{p}^{-1}}.\label{eq_A}%
\end{equation}

By using that $-n<\gamma<0$,%
\begin{gather}
p^{\gamma}\sum_{l=0}^{\infty}p^{nl}p^{\gamma l}e^{-C_{1}p^{1+l}\left\Vert
\xi\right\Vert _{p}^{-1}}\leq\sum_{l=0}^{\infty}p^{nl}e^{-C_{1}p^{1+l}%
\left\Vert \xi\right\Vert _{p}^{-1}}\nonumber\\
=e^{-C_{1}p\left\Vert \xi\right\Vert _{p}^{-1}}\sum_{l=0}^{\infty}%
p^{nl}e^{-C_{1}p\left\Vert \xi\right\Vert _{p}^{-1}\left(  p^{l}-1\right)
}=e^{-C_{1}p\left\Vert \xi\right\Vert _{p}^{-1}}\left\{  1+\sum_{l=1}^{\infty
}p^{nl}e^{-C_{1}p\left\Vert \xi\right\Vert _{p}^{-1}\left(  p^{l}-1\right)
}\right\}  .\label{Eq_B}%
\end{gather}
By using that $p^{l}-1\geq p^{l-1}$ for any positive integer,
\begin{gather}
e^{-C_{1}p\left\Vert \xi\right\Vert _{p}^{-1}}\left\{  1+\sum_{l=1}^{\infty
}p^{nl}e^{-C_{1}p\left\Vert \xi\right\Vert _{p}^{-1}\left(  p^{l}-1\right)
}\right\}  \leq e^{-C_{1}p\left\Vert \xi\right\Vert _{p}^{-1}}\left\{
1+\sum_{l=1}^{\infty}p^{nl}e^{-C_{1}p^{l}\left\Vert \xi\right\Vert _{p}^{-1}%
}\right\} \nonumber\\
=e^{-C_{1}p\left\Vert \xi\right\Vert _{p}^{-1}}\left\{  1+\frac{1}{1-p^{-n}%
}\int\limits_{\left\Vert y\right\Vert _{p}>1}e^{-C_{1}\left\Vert y\right\Vert
_{p}\left\Vert \xi\right\Vert _{p}^{-1}}d^{n}y\right\} \nonumber\\
\leq e^{-C_{1}p\left\Vert \xi\right\Vert _{p}^{-1}}\left\{  1+\frac
{1}{1-p^{-n}}\int\limits_{\mathbb{Q}_{p}^{n}}e^{-C_{1}\left\Vert y\right\Vert
_{p}\left\Vert \xi\right\Vert _{p}^{-1}}d^{n}y\right\}  \leq e^{-C_{1}%
p\left\Vert \xi\right\Vert _{p}^{-1}}\left\{  1+A_{0}\left\Vert \xi\right\Vert
_{p}^{n}\right\}  ,\label{Eq_C}%
\end{gather}
where we used the well-known estimation%
\[
\int\limits_{\mathbb{Q}_{p}^{n}}e^{-\tau\left\Vert y\right\Vert _{p}}%
d^{n}y\leq C_{0}\tau^{-n}\text{ for }\tau>0.
\]
Therefore from (\ref{eq_A})-(\ref{Eq_C}),%
\[
1-\widehat{J}(\left\Vert \xi\right\Vert _{p})\leq B_{1}\left\Vert
\xi\right\Vert _{p}^{-n-\gamma}e^{-C_{1}p\left\Vert \xi\right\Vert _{p}^{-1}%
}+B_{2}\left\Vert \xi\right\Vert _{p}^{-\gamma}e^{-C_{1}p\left\Vert
\xi\right\Vert _{p}^{-1}}\text{ for }\xi\neq0.
\]

\end{proof}

\begin{remark}
\label{Nota_support}We notice that the fact that $J$ is of exponential type
implies that supp $J\left(  \left\Vert x\right\Vert _{p}\right)
\nsubseteqq\mathbb{Z}_{p}^{n}$. Then $1-\widehat{J}(1)>0$. Indeed, by Lemma
\ref{lemma1}-(ii),%
\begin{align*}
1-\widehat{J}(1)  & =J(p)+p^{n}\sum\nolimits_{l=0}^{\infty}p^{nl}%
J(p^{1+l})=J(p)+\sum\nolimits_{k=1}^{\infty}p^{nk}J(p^{k})\\
& =J(p)+\frac{1}{1-p^{-n}}\int\nolimits_{%
\mathbb{Q}
_{p}^{n}\smallsetminus\mathbb{Z}_{p}^{n}}J\left(  \left\Vert x\right\Vert
_{p}\right)  d^{n}x.
\end{align*}
If $1-\widehat{J}(1)=0$, then $J(p)=0$ and $J\left(  \left\Vert x\right\Vert
_{p}\right)  \equiv0$ for $x\in%
\mathbb{Q}
_{p}^{n}\smallsetminus\mathbb{Z}_{p}^{n}$, i.e. supp $J\left(  \left\Vert
x\right\Vert _{p}\right)  \subseteq\mathbb{Z}_{p}^{n}$.
\end{remark}

\begin{lemma}
\label{lemma2A}If $-n<\gamma<0$ and $J$ is of exponential type, then
\[
\frac{\Omega\left(  \left\Vert \xi\right\Vert _{p}\right)  }{1-\widehat
{J}\left(  \left\Vert \xi\right\Vert _{p}\right)  }\notin L^{1}\left(
\mathbb{Q}
_{p}^{n},d^{n}\xi\right)  .
\]

\end{lemma}

\begin{proof}
By using Lemma \ref{lemma1}-(iii),
\begin{gather*}
\int\nolimits_{\mathbb{Z}_{p}^{n}}\frac{d^{n}\xi}{1-\widehat{J}\left(
||\xi||_{p}\right)  }\geq\int\nolimits_{\mathbb{Z}_{p}^{n}}\frac{d^{n}\xi
}{B_{1}\left\Vert \xi\right\Vert _{p}^{-n-\gamma}e^{-C_{1}p\left\Vert
\xi\right\Vert _{p}^{-1}}+B_{2}\left\Vert \xi\right\Vert _{p}^{-\gamma
}e^{-C_{1}p\left\Vert \xi\right\Vert _{p}^{-1}}}\\
=\int\nolimits_{\mathbb{Z}_{p}^{n}}\frac{\left\Vert \xi\right\Vert
_{p}^{\gamma}e^{C_{1}p\left\Vert \xi\right\Vert _{p}^{-1}}d^{n}\xi}%
{B_{1}\left\Vert \xi\right\Vert _{p}^{-n}+B_{2}}=\sum\limits_{j=0}^{\infty
}\frac{p^{-j\gamma}e^{C_{1}p^{j+1}}}{B_{1}p^{jn}+B_{2}}%
{\displaystyle\int\limits_{||\xi||_{p}=p^{-j}}}
d^{n}\xi\\
=\left(  1-p^{-n}\right)  \sum\limits_{j=0}^{\infty}\frac{p^{-j\left(
n+\gamma\right)  }e^{C_{1}p^{j+1}}}{B_{1}p^{jn}+B_{2}}=\infty.
\end{gather*}

\end{proof}

\begin{remark}
\label{Nota1}

\noindent(i) A function $f:%
\mathbb{Q}
_{p}^{n}\rightarrow%
\mathbb{C}
$ is called positive definite, if
\[%
{\textstyle\sum\nolimits_{i,j=1}^{m}}
f(x_{i}-x_{j})\lambda_{i}\overline{\lambda_{j}}\geq0
\]
for all $m\in\mathbb{N}\backslash\{0\}$ , $x_{1},\ldots,x_{m}$ $\in$
$\mathbb{Q}_{p}^{n}$ and $\lambda_{1},\ldots,\lambda_{m}\in\mathbb{C}$. By a
direct calculation one verifies that\textbf{\ }$\widehat{J}(\left\Vert
\xi\right\Vert _{p})$ is a positive definite function.

\noindent(ii) A function $f:\mathbb{%
\mathbb{Q}
}_{p}^{n}\rightarrow%
\mathbb{C}
$ is called negative definite, if$\ $%
\[%
{\textstyle\sum\nolimits_{i,j=1}^{m}}
\left(  f(x_{i})+\overline{f(x_{j})}-f(x_{i}-x_{j})\right)  \lambda
_{i}\overline{\lambda_{j}}\geq0
\]
for all $m\in\mathbb{N}\backslash\{0\},$ $x_{1},\ldots,x_{m}\in$ $\mathbb{%
\mathbb{Q}
}_{p}^{n},$ $\lambda_{1},\ldots,\lambda_{m}$ $\in$ $\mathbb{C}$. A theorem due
to Schoenberg asserts that a function $f:%
\mathbb{Q}
_{p}^{n}\rightarrow%
\mathbb{C}
$ is negative definite if and only if the following two conditions are
satisfied:(1) $f(0)\geq0$ and $\ $(2) the function $x\longmapsto e^{-tf(x)}$
is positive definite for all $t>0$, cf. \cite[Theorem 7.8]{Berg-Gunnar}. By
using Corollary 7.7 in \cite[Theorem 7.8]{Berg-Gunnar}, the function
$\widehat{J}(0)-\widehat{J}(\left\Vert \xi\right\Vert _{p})=1-\widehat
{J}(\left\Vert \xi\right\Vert _{p})$ is negative definite, by applying
Schoenberg Theorem, the function $e^{t(\widehat{J}(\left\Vert \xi\right\Vert
_{p})-1)}$ is positive definite for all $t>0.$
\end{remark}

\subsection{\label{Section_Laplacians}A class of nonlocal $p$-adic Operators}

We define the operator $Af=J\ast f-f$ with $J$ as in Section
\ref{Section_Exp_Lan}. Then, for any $1\leq\rho\leq\infty$, $A:L^{\rho
}\longrightarrow L^{\rho}$ gives rise a well-defined linear bounded operator.
Indeed, by the Young inequality
\[
||Af||_{L^{\rho}}\leq||J\ast f||_{L^{\rho}}+||f||_{L^{\rho}}\leq||J||_{L^{1}%
}||f||_{L^{\rho}}+||f||_{L^{\rho}}\leq2||f||_{L^{\rho}}.
\]

\begin{proposition}
\label{lemma_prop1} Consider $A:L^{2}\left(  \mathbb{Q}_{p}^{n}\right)
\rightarrow L^{2}\left(  \mathbb{Q}_{p}^{n}\right)  $ given by
\[
Af\left(  x\right)  =\mathcal{F}_{\xi\rightarrow x}^{-1}\left(  (\widehat
{J}(||\xi||_{p})-1)\mathcal{F}_{x\rightarrow\xi}f\right)  ,
\]
and \ the Cauchy problem :%
\begin{equation}
\left\{
\begin{array}
[c]{ll}%
\frac{\partial u}{\partial t}(x,t)=Au(x,t)\text{,} & t\in\left[
0,\infty\right)  \text{,\ }x\in\mathbb{Q}_{p}^{n}\\
& \\
u(x,0)=u_{0}(x)\in\mathcal{D}(\mathbb{Q}_{p}^{n})\text{.} &
\end{array}
\right. \label{Cauchy_problem}%
\end{equation}
Then%
\[
u(x,t)=\int_{\mathbf{%
\mathbb{Q}
}_{p}^{n}}\chi_{p}\left(  -\xi\cdot x\right)  e^{(\widehat{J}(||\xi
||_{p})-1)t}\widehat{u_{0}}(\xi)d^{n}\xi
\]
is a classical solution of (\ref{Cauchy_problem}). In addition, $u(\cdot,t)$
is a continuous function for any $t\geq0$.
\end{proposition}

\begin{proof}
The result follows from the following assertions.

\textbf{Claim 1. }$u(x,\cdot)\in C^{1}\left(  \left[  0,\infty\right)
\right)  $ and
\[
\frac{\partial u}{\partial t}(x,t)=\int_{\mathbf{%
\mathbb{Q}
}_{p}^{n}}\chi_{p}\left(  -\xi\cdot x\right)  (\widehat{J}(||\xi
||_{p})-1)e^{(\widehat{J}(||\xi||_{p})-1)t}\widehat{u_{0}}(\xi)d^{n}\xi
\]
for $t\geq0$,\ $x\in\mathbb{Q}_{p}^{n}$.

The formula follows from the fact that%
\[
\left\vert \chi_{p}\left(  -\xi\cdot x\right)  e^{(\widehat{J}(||\xi
||_{p})-1)t}\widehat{u_{0}}(\xi)\right\vert \leq\left\vert \widehat{u_{0}}%
(\xi)\right\vert \in L^{1}%
\]
and that
\[
\left\vert \chi_{p}\left(  -\xi\cdot x\right)  (\widehat{J}(||\xi
||_{p})-1)e^{(\widehat{J}(||\xi||_{p})-1)t}\widehat{u_{0}}(\xi)\right\vert
\leq2\left\vert \widehat{u_{0}}(\xi)\right\vert \in L^{1},
\]

cf. Lemma \ref{lemma1}-(i), by applying the Dominated Convergence Theorem.

\textbf{Claim 2. }%
\[
Au(x,t)=\int_{\mathbf{%
\mathbb{Q}
}_{p}^{n}}\chi_{p}\left(  -\xi\cdot x\right)  (\widehat{J}(||\xi
||_{p})-1)e^{(\widehat{J}(||\xi||_{p})-1)t}\widehat{u_{0}}(\xi)d^{n}\xi
\]
for $t\in\left[  0,\infty\right)  $,\ $x\in\mathbb{Q}_{p}^{n}$.

The \ formula follows from the fact that $u(x,t)=\mathcal{F}_{\xi\rightarrow
x}^{-1}\left(  e^{(\widehat{J}(||\xi||_{p})-1)t}\widehat{u_{0}}(\xi)\right)
\in L^{2}\left(  \mathbb{Q}_{p}^{n}\right)  $ for any $t\geq0$ and that
$\left(  \widehat{J}(||\xi||_{p})-1\right)  e^{(\widehat{J}(||\xi||_{p}%
)-1)t}\widehat{u_{0}}(\xi)\in L^{2}\left(  \mathbb{Q}_{p}^{n}\right)  $ for
any $t\geq0$, cf. Lemma \ref{lemma1}-(i).
\end{proof}

\section{Heat Kernels}

We define the \textit{heat Kernel} attached to operator $A$ as%
\[
Z(x,t)=\mathcal{F}_{\xi\rightarrow x}^{-1}(e^{(\widehat{J}(||\xi||_{p}%
)-1)t})\in\mathcal{D}^{\prime}(\mathbb{Q}_{p}^{n})\text{ for }t\geq0.
\]
When considering $Z(x,t)$ as a function of $x$ for $t$ fixed we will write
$Z_{t}(x).$

We recall that a distribution $F\in\mathcal{D}^{\prime}(\mathbb{Q}_{p}^{n}) $
is called \textit{positive}, if $(F,\varphi)\geq0$ for every \textit{positive
test function} $\varphi$, i.e. for $\varphi:\mathbb{Q}_{p}^{n}\longrightarrow%
\mathbb{R}
_{+}$, $\varphi\in\mathcal{D}(\mathbb{Q}_{p}^{n})$.

A distribution $F$ is \textit{positive definite}, if for every test function
$\varphi,$ the inequality $(F,\overline{\varphi\ast\widetilde{\varphi}})\geq0
$ holds, where $\widetilde{\varphi}(x)=\overline{\varphi(-x)}$ and
$\ \overline{\varphi(-x)}$ denotes the complex conjugate of $\varphi(-x)$.

\begin{theorem}
[{$p-$adic Bochner-Schwartz Theorem \cite[Theorem 4.1]{Z1}}]\label{Theorem1}
Every positive-definite distribution $F$ on $\mathbb{Q}_{p}^{n}$ is the
Fourier transform of a regular Borel measure $\mu$ on $\mathbb{Q}_{p}^{n}$,
i.e.%
\[
(F,\varphi)=\int_{\mathbf{%
\mathbb{Q}
}_{p}^{n}}\widehat{\varphi}(\xi)d\mu(\xi),\text{ for }\varphi\in
\mathcal{D}(\mathbb{Q}_{p}^{n}).
\]

\end{theorem}

\begin{remark}
\label{Nota2}If $f:%
\mathbb{Q}
_{p}^{n}\longrightarrow%
\mathbb{C}
$ is a continuous positive-definite function, then%
\[
(f,\overline{\varphi\ast\widetilde{\varphi}})\geq0\text{ \ for }\varphi
\in\mathcal{D}(\mathbb{Q}_{p}^{n}),
\]
which means that $f$ generates a positive-definite distribution , see e.g.
\cite[Proposition 4.1]{Berg-Gunnar}.
\end{remark}

\begin{lemma}
\label{lemma2} Let $\varphi$ be a positive test function. Then
\[
\int_{\mathbf{%
\mathbb{Q}
}_{p}^{n}}\chi_{p}\left(  -\xi\cdot x\right)  e^{(\widehat{J}(\left\Vert
\xi\right\Vert _{p})-1)t}\widehat{\varphi}(\xi)d^{n}\xi=\left(  Z_{t}%
\ast\varphi\right)  (x)\geq0\text{ \ for }x\in\mathbb{Q}_{p}^{n}\text{ and
}t\geq0.
\]

\end{lemma}

\begin{proof}
It is sufficient to show the lemma for $x\in\mathbb{Q}_{p}^{n}$ and $t>0$. By
Remark \ref{Nota1}-(ii), the function $e^{t(\widehat{J}(\left\Vert
\xi\right\Vert _{p})-1)}$ is positive definite for all $t>0$, by Remark
\ref{Nota2} and Theorem \ref{Theorem1}, $Z_{t}(x)=\mathcal{F}_{\xi\rightarrow
x}^{-1}(e^{(\widehat{J}(\left\Vert \xi\right\Vert _{p})-1)t})$ is a Borel
measure on $%
\mathbb{Q}
_{p}^{n}$ for $t>0$, which is identified with a positive distribution. Then
\begin{align*}
(e^{(\widehat{J}(\left\Vert \xi\right\Vert _{p})-1)t},\chi_{p}\left(
-\xi\cdot x\right)  \widehat{\varphi}(\xi)) &  =\left(  \mathcal{F}%
_{\xi\rightarrow y}^{-1}(e^{(\widehat{J}(\left\Vert \xi\right\Vert _{p}%
)-1)t}),\mathcal{F}_{\xi\rightarrow y}(\chi_{p}\left(  -\xi\cdot x\right)
\widehat{\varphi}(\xi))\right) \\
&  =(Z_{t}(y),\varphi(x-y))\geq0\text{ for }t>0\text{,}%
\end{align*}
since $\varphi(x-y)\geq0$.
\end{proof}

\subsection{Decaying of the heat kernel at infinity}

Let $h\left(  \left\Vert \xi\right\Vert _{p}\right)  \in L_{\text{loc}}^{1}$,
then%
\[
\sum\limits_{j=-m}^{m}h\left(  p^{j}\right)  1_{S_{j}^{n}}\left(  \xi\right)
\text{ }\rightarrow h\left(  \left\Vert \xi\right\Vert _{p}\right)  \text{ in
}\mathcal{D}^{\prime}(\mathbb{Q}_{p}^{n})\text{.}%
\]
Now, by using \cite[Theorem 4.9.3]{Alberio et al} and the fact that
$\mathcal{F}^{-1}$ is continuous on $\mathcal{D}^{\prime}(\mathbb{Q}_{p}^{n})
$,%
\begin{gather*}
\mathcal{F}_{\xi\rightarrow x}^{-1}\left(  \sum\limits_{j=-m}^{m}h\left(
p^{j}\right)  1_{S_{j}^{n}}\left(  \xi\right)  \right)  =\sum\limits_{j=-m}%
^{m}h\left(  p^{j}\right)  \mathcal{F}_{\xi\rightarrow x}^{-1}\left(
1_{S_{j}^{n}}\left(  \xi\right)  \right)  \text{ }\\
=\sum\limits_{j=-m}^{m}h\left(  p^{j}\right)  \int\limits_{S_{j}^{n}}\chi
_{p}\left(  -x\cdot\xi\right)  d^{n}\xi\text{\ in }\mathcal{D}^{\prime
}(\mathbb{Q}_{p}^{n}),
\end{gather*}
therefore%
\begin{equation}
\mathcal{F}_{\xi\rightarrow x}^{-1}\left(  h\left(  \left\Vert \xi\right\Vert
_{p}\right)  \right)  =\sum\limits_{j=-\infty}^{\infty}h\left(  p^{j}\right)
\int\limits_{S_{j}^{n}}\chi_{p}\left(  -x\cdot\xi\right)  d^{n}\xi\text{\ in
}\mathcal{D}^{\prime}(\mathbb{Q}_{p}^{n}).\label{Identity_F}%
\end{equation}
Suppose now that $\sum_{k=0}^{\infty}p^{-kn}h\left(  p^{-k}\left\Vert
x\right\Vert _{p}^{-1}\right)  <\infty$, then%
\begin{gather*}
\widetilde{h}\left(  x\right)  :=\int\limits_{\mathbb{Q}_{p}^{n}}\chi
_{p}\left(  -x\cdot\xi\right)  h\left(  \left\Vert \xi\right\Vert _{p}\right)
d^{n}\xi:=\lim_{m\rightarrow\infty}\sum\limits_{j=-m}^{m}\text{ }%
\int\limits_{S_{j}^{n}}\chi_{p}\left(  -x\cdot\xi\right)  h\left(  \left\Vert
\xi\right\Vert _{p}\right)  d^{n}\xi\\
=\sum\limits_{j=-\infty}^{\infty}h\left(  p^{j}\right)  \int\limits_{S_{j}%
^{n}}\chi_{p}\left(  -x\cdot\xi\right)  d^{n}\xi\\
=\left(  1-p^{-n}\right)  \left\Vert x\right\Vert _{p}^{-n}\sum\limits_{k=0}%
^{\infty}p^{-kn}h\left(  p^{-k}\left\Vert x\right\Vert _{p}^{-1}\right)
-\left\Vert x\right\Vert _{p}^{-n}h\left(  p\left\Vert x\right\Vert _{p}%
^{-1}\right)  \text{ for }x\neq0,
\end{gather*}
and by comparing with (\ref{Identity_F}), we get%
\[
\left(  \mathcal{F}_{\xi\rightarrow x}^{-1}\left[  h\left(  \left\Vert
\xi\right\Vert _{p}\right)  \right]  ,\phi\left(  x\right)  \right)  =\left(
\widetilde{h}\left(  x\right)  ,\phi\left(  x\right)  \right)
\]
for $\phi\in\mathcal{D}(\mathbb{Q}_{p}^{n})$ with supp $\phi\subset%
\mathbb{Q}
_{p}^{n}\backslash\{0\}$. It is important to highlight that function
$\widetilde{h}\left(  x\right)  $ is not the inverse Fourier transform of
$h\left(  \left\Vert \xi\right\Vert _{p}\right)  $ in the classical sense,
because we use an improper integral in its definition. This is the reason for
the `new notation'. We formally summarize the above reasoning in the following lemma:

\begin{lemma}
Let $h\left(  \left\Vert \xi\right\Vert _{p}\right)  \in L_{\text{loc}}^{1}$
satisfying $\sum_{k=0}^{\infty}p^{-kn}h\left(  p^{-k}\left\Vert \xi\right\Vert
_{p}^{-1}\right)  <\infty$ for $\xi\neq0$, then $\mathcal{F}_{\xi\rightarrow
x}^{-1}\left[  h\left(  \left\Vert \xi\right\Vert _{p}\right)  \right]
=\widetilde{h}\left(  x\right)  $ as a distribution on $\mathcal{D}(%
\mathbb{Q}
_{p}^{n}\backslash\{0\})$.
\end{lemma}

We now apply this lemma to the case $h\left(  \left\Vert \xi\right\Vert
_{p}\right)  =e^{(\widehat{J}(\left\Vert \xi\right\Vert _{p})-1)t}$, with
$t\geq0$:%
\[
\left(  Z\left(  x,t\right)  ,\phi\left(  x\right)  \right)  =\left(
\widetilde{Z}\left(  x,t\right)  ,\phi\left(  x\right)  \right)
\]
for $\phi\in\mathcal{D}(\mathbb{Q}_{p}^{n})$ with supp $\phi\subset%
\mathbb{Q}
_{p}^{n}\backslash\{0\}$, where%
\begin{equation}
\widetilde{Z}\left(  x,t\right)  =||x||_{p}^{-n}\left[  (1-p^{-n})\sum
_{j=0}^{\infty}p^{-nj}e^{(\widehat{J}(p^{-j}\left\Vert x\right\Vert _{p}%
^{-1})-1)t}-e^{(\widehat{J}(p||x||_{p}^{-1})-1)t}\right] \label{Zeta_mono}%
\end{equation}
for $t\geq0$ and $x\neq0$.

\begin{proposition}
\label{prop3}Assume that $J$ is of exponential type, then the following
estimations hold:

\noindent(i) $\widetilde{Z}(x,t)\leq2t||x||_{p}^{-n}$, for $x\in%
\mathbb{Q}
_{p}^{n}\backslash\{0\}$ and $t>0$;

\noindent(ii) if $-n<\gamma<0$, then \ $\widetilde{Z}(x,t)\leq C_{0}%
t\left\Vert x\right\Vert _{p}^{\gamma}e^{-C_{1}\left\Vert x\right\Vert _{p}}$,
for $\left\Vert x\right\Vert _{p}>p^{l}$, $l\in\mathbb{Z}$, and $t>0 $, where
the positive constant $C_{0}$\ depends on $l$.
\end{proposition}

\begin{proof}
The estimations follow from the following Claim:

\textbf{Claim.} $\widetilde{Z}(x,t)\leq t||x||_{p}^{-n}\left\{  1-\widehat
{J}(p||x||_{p}^{-1})\right\}  $, for $x\in%
\mathbb{Q}
_{p}^{n}\backslash\{0\}$ and $t>0$.

We notice that by Lemma \ref{lemma1}-(ii) $1-\widehat{J}(p||x||_{p}^{-1}%
)\geq0$. The first estimation follows from Lemma \ref{lemma1}-(i) and the
Claim. The second estimation follows from the Claim and Lemma \ref{lemma1}-(iii).

The proof of the Claim is as follows. By using that $e^{(\widehat{J}%
(p^{-j}\left\Vert x\right\Vert _{p}^{-1})-1)t}\leq1$ for $j\in%
\mathbb{N}
$, cf. Lemma \ref{lemma1}-(i), we get that%
\begin{gather*}
\widetilde{Z}(x,t)\leq||x||_{p}^{-n}\left[  (1-p^{-n})\sum_{j\geq0}%
p^{-nj}-e^{(\widehat{J}(p||x||_{p}^{-1})-1)t}\right] \\
=||x||_{p}^{-n}\left[  \sum_{j\geq0}(p^{-nj}-p^{-n(j+1)})-e^{(\widehat
{J}(p||x||_{p}^{-1})-1)t}\right] \\
=||x||_{p}^{-n}\left\{  1-e^{(\widehat{J}(p||x||_{p}^{-1})-1)t}\right\}  .
\end{gather*}
We now apply the Mean-Value Theorem to the real function $e^{(\widehat
{J}(p||x||_{p}^{-1})-1)u}$ on $[0,t]$ with $t>0$,
\[
e^{(\widehat{J}(p||x||_{p}^{-1})-1)t}-1=\left\{  \widehat{J}(p||x||_{p}%
^{-1})-1\right\}  te^{(\widehat{J}(p||x||_{p}^{-1})-1)\tau}%
\]
for some $\tau\in(0,t)$, consequently $1-e^{(\widehat{J}(p||x||_{p}^{-1}%
)-1)t}\leq\left\{  1-\widehat{J}(p||x||_{p}^{-1})\right\}  t$. Hence,%
\[
\widetilde{Z}(x,t)\leq t||x||_{p}^{-n}\left\{  1-\widehat{J}(p||x||_{p}%
^{-1})\right\}  .
\]

\end{proof}

\section{L\'{E}VY PROCESSES}

For the basic results on Hunt, L\'{e}vy and Markov processes the reader may
consult \cite{Dyn}, \cite{Taira}, \cite{Blumenthal-Getoor}, \cite{Evans}.

\begin{remark}
\label{Nota3}We denote by $\mathcal{B}_{0}$ the family of subsets of
$\ \mathbf{%
\mathbb{Q}
}_{p}^{n}$ formed by finite unions of disjoint balls and the empty set. This
family has a natural structure of Boolean ring, i.e. if $B_{1}$, $B_{2}%
\in\mathcal{B}_{0}$ then $B_{1}\cup B_{2}\in\mathcal{B}_{0}$ and
$B_{1}\backslash B_{2}\in\mathcal{B}_{0}$. The Caratheodory Theorem asserts
that if $\mu$ is a $\sigma-$finite measure on $\mathcal{B}_{0}$, then there is
a unique measure also denoted by $\mu$ on $\mathcal{B}(%
\mathbb{Q}
_{p}^{n})$, the $\sigma-$ring generated by $\mathcal{B}_{0}$, which is
$\sigma-$ring of Borel sets of $\mathbf{%
\mathbb{Q}
}_{p}^{n}$, see \cite[Theorem A, p. 54]{Halmos}. Then every positive
distribution can be identified with a Borel measure on $\mathbf{%
\mathbb{Q}
}_{p}^{n}$.
\end{remark}

\begin{definition}
For $E\in\mathcal{B}_{0}(%
\mathbb{Q}
_{p}^{n})$, we define%
\[
p_{t}(x,E)=\left\{
\begin{array}
[c]{ll}%
Z_{t}(x)\ast1_{E}(x)\text{,} & \text{\ for }t>0\text{, }x\in\mathbf{%
\mathbb{Q}
}_{p}^{n}\\
& \\
1_{E}(x), & \text{for }t=0\text{, }x\in\mathbf{%
\mathbb{Q}
}_{p}^{n}.
\end{array}
\right.
\]

\end{definition}

\begin{lemma}
\label{lemma3} $p_{t}(x,\cdot)$, $t\geq0$, $x\in\mathbf{%
\mathbb{Q}
}_{p}^{n}$, is a measure on $\mathcal{B}(%
\mathbb{Q}
_{p}^{n})$.
\end{lemma}

\begin{proof}
By Lemma \ref{lemma2} and the fact that
\[
p_{t}(x,E)=\int_{\mathbf{%
\mathbb{Q}
}_{p}^{n}}\chi_{p}\left(  -\xi\cdot x\right)  e^{(\widehat{J}(\left\Vert
\xi\right\Vert _{p})-1)t}\text{ }\widehat{1_{E}}(\xi)d^{n}\xi\text{ for
}t>0\text{, }x\in\mathbf{%
\mathbb{Q}
}_{p}^{n}%
\]
$p_{t}(x,E)$ is a measure on $\mathcal{B}_{0}$, in addition, $p_{t}(x,\cdot) $
has a unique extension to a measure on the Borel $\sigma-$ring of $\mathbf{%
\mathbb{Q}
}_{p}^{n}$. We denote this extension also by $p_{t}(x,\cdot)$. Indeed, by the
Caratheodory Theorem, see Remark \ref{Nota3}, it is sufficient to show that
$\mathbf{%
\mathbb{Q}
}_{p}^{n}$ is a countable disjoint union of balls $B_{i}^{n}$, $i\in%
\mathbb{N}
$, satisfying $p_{t}(x,B_{i}^{n})<\infty$ for any $i\in%
\mathbb{N}
$. Indeed,%
\[%
\mathbb{Q}
_{p}^{n}=\bigsqcup\limits_{\widetilde{x}_{i}\in(%
\mathbb{Q}
_{p}/%
\mathbb{Z}
_{p})^{n}}B_{0}^{n}(\widetilde{x}_{i}),
\]
where the elements of $%
\mathbb{Q}
_{p}/%
\mathbb{Z}
_{p}$ have the form $\widetilde{y}=a_{-m}p^{-m}+\cdots+a_{-1}p^{-1}$ with
$a_{i}\in\{0,\cdots,p-1\}$. The correspondence $\widetilde{y}\mapsto
a_{m}p^{m}+\cdots+a_{1}p^{1}$ implies that $%
\mathbb{Q}
_{p}/%
\mathbb{Z}
_{p}$ is countable, and therefore $(%
\mathbb{Q}
_{p}/%
\mathbb{Z}
_{p})^{n}$ is also countable. Finally, the condition $p_{t}(x,B_{0}%
^{n}(\widetilde{x}_{i}))<\infty$ follows from the fact that $\widehat
{1_{B_{0}^{n}(\widetilde{x}_{i})}}(\xi)$ has compact support.
\end{proof}

\begin{proposition}
\label{prop2} $p_{t}(x,E)$ for $t\geq0,$ $x\in%
\mathbb{Q}
_{p}^{n}$, $E\in\mathcal{B}(%
\mathbb{Q}
_{p}^{n})$ is a Markov transition function on $%
\mathbb{Q}
_{p}^{n}$.
\end{proposition}

\begin{proof}
\textbf{Claim 1.} $p_{t}(x,\cdot)$ is measure on $\mathcal{B}(%
\mathbb{Q}
_{p}^{n})$ and $p_{t}(x,%
\mathbb{Q}
_{p}^{n})=1$ for all $t\geq0$ and $x\in%
\mathbb{Q}
_{p}^{n}.$

The first part of the assertion was established in Lemma \ref{lemma3}. To show
the second part of the assertion, we\ notice that $B_{k}^{n}$, $k\in%
\mathbb{N}
$, is an increasing sequence of Borelian sets converging to $%
\mathbb{Q}
_{p}^{n},$ i.e. $B_{k}^{n}\uparrow%
\mathbb{Q}
_{p}^{n}$. Set $\Delta_{k}(x):=\Omega(p^{-k}||x||_{p}),$ $k\in%
\mathbb{N}
$, hence $p_{t}(x,%
\mathbb{Q}
_{p}^{n})=\lim_{k\rightarrow\infty}p_{t}(x,\Delta_{k})$. Now, since
\[
\widehat{\Delta}_{k}(\xi)=\delta_{k}(\xi)=p^{kn}\left\{
\begin{array}
[c]{ccc}%
1 & \text{if} & ||\xi||_{p}\leq p^{-k}\\
&  & \\
0 & \text{if} & ||\xi||_{p}>p^{-k},
\end{array}
\right.
\]%
\begin{align*}
p_{t}(x,%
\mathbb{Q}
_{p}^{n}) &  =\lim_{k\rightarrow\infty}\int_{\mathbf{%
\mathbb{Q}
}_{p}^{n}}\chi_{p}\left(  -\xi\cdot x\right)  e^{(\widehat{J}(||\xi
||_{p})-1)t}\delta_{k}(\xi)d^{n}\xi\\
&  =\lim_{k\rightarrow\infty}p^{nk}\int_{||\xi||_{p}\leq p^{-k}}\chi
_{p}\left(  -\xi\cdot x\right)  e^{(\widehat{J}(||\xi||_{p})-1)t}d^{n}%
\xi\text{ \ (taking }p^{-k}\xi=z\text{)}\\
&  =\lim_{k\rightarrow\infty}\int_{||z||_{p}\leq1}\chi_{p}\left(  -p^{k}z\cdot
x\right)  e^{(\widehat{J}(p^{-k}||z||_{p})-1)t}d^{n}z\\
&  =\lim_{k\rightarrow\infty}\int_{||z||_{p}\leq1}e^{(\widehat{J}%
(p^{-k}||z||_{p})-1)t}d^{n}z\text{, for }k\text{ big enough,}%
\end{align*}
because for $k$ big enough $p^{k}x\in%
\mathbb{Z}
_{p}^{n}$ and thus $\chi_{p}\left(  -p^{k}z\cdot x\right)  \equiv1$. By using
that $e^{(\widehat{J}(p^{-k}||z||_{p})-1)t}\leq1$ for any $t\geq0,$ and that
$\lim_{k\rightarrow\infty}e^{(\widehat{J}(p^{-k}||z||_{p})-1)t}=1$
($\widehat{J}$ is continuous at the origin and $\widehat{J}\left(  0\right)
=1 $), and by applying the Dominated Convergence Theorem,%
\[
p_{t}(x,%
\mathbb{Q}
_{p}^{n})=\int_{||z||_{p}\leq1}d^{n}z=1.
\]

\textbf{Claim 2.} $p_{t}(\cdot,E)$ is a Borel measurable function for all
$t>0$ and $E\in\mathcal{B}\left(  \mathbf{%
\mathbb{Q}
}_{p}^{n}\right)  $.

Define $E_{k}=\Delta_{k}E$, $k\in%
\mathbb{N}
$, then $E_{k}\uparrow E$ with $E_{k}\in\mathcal{B}\left(  \mathbf{%
\mathbb{Q}
}_{p}^{n}\right)  $. By abuse of language, we use the notation $p_{t}(x,E_{k})
$ to mean a function of $(t,x)$ with $E_{k}$ fixed. Now, $p_{t}(x,E_{k})$ is
the solution of%
\[
\left\{
\begin{array}
[c]{ll}%
\frac{\partial}{\partial t}p_{t}(x,E_{k})=J\left(  x\right)  \ast
p_{t}(x,E_{k})-p_{t}(x,E_{k}), & t\in\lbrack0,\infty),\text{\ }x\in\mathbf{%
\mathbb{Q}
}_{p}^{n}\\
& \\
p_{0}(x,E_{k})=1_{E_{k}}, & 1_{E_{k}}\in L^{1}(\mathbf{%
\mathbb{Q}
}_{p}^{n}),
\end{array}
\right.
\]
cf. Proposition \ref{lemma_prop1}. Then ,\textbf{\ }$p_{t}(x,E_{k})$ is a
continuous function in $x$ for any $t\geq0$, which implies that $p_{t}%
(\cdot,E_{k})$ is a measurable function of $x$ for any $t\geq0$. Now, by using
that $E_{k}\uparrow E$ and the fact $p_{t}(x,\cdot)$ is a measure on
$\mathcal{B}(%
\mathbb{Q}
_{p}^{n})$, we get that $p_{t}(x,E_{k})\rightarrow p_{t}(x,E)$ as
$k\rightarrow\infty$, which implies that $p_{t}(\cdot,E)$ is the pointwise
limit of a sequence of measurable functions $\left\{  p_{t}(\cdot
,E_{k})\right\}  _{k}$ and consequently $p_{t}(\cdot,E)$ is measurable.

\textbf{Claim 3.} $p_{0}(x,\{x\})=1$ for all $x\in\mathbf{%
\mathbb{Q}
}_{p}^{n}.$

This is a direct consequence of the definition of measure $p_{t}(x,E)$.

\textbf{Claim 4. (The Chapman-Kolmogorov equation)} For all $t,s\geq0,$
$x\in\mathbf{%
\mathbb{Q}
}_{p}^{n}$ and $E\in\mathcal{B}(\mathbb{Q}_{p}^{n})$,%
\[
p_{t+s}(x,E)=\int_{\mathbf{%
\mathbb{Q}
}_{p}^{n}}p_{t}(x,d^{n}y)p_{s}(y,E).
\]
We consider the case $t$, $s>0$, since in the other cases the assertion is
clear. We first note that%
\begin{equation}
p_{t+s}(x,\cdot)=p_{t}(x,\cdot)\ast p_{s}(x,\cdot)\text{ in }\mathcal{D}%
^{\prime}(%
\mathbb{Q}
_{p}^{n}).\label{Ch_K_1}%
\end{equation}
Indeed, for $E\in\mathcal{B}_{0}$,%
\begin{align*}
p_{t+s}(x,E) &  =\mathcal{F}_{\xi\rightarrow x}^{-1}(e^{(\widehat{J}%
(||\xi||_{p})-1)(t+s)})\ast1_{E}=\mathcal{F}_{\xi\rightarrow x}^{-1}%
(e^{(\widehat{J}(||\xi||_{p})-1)t}\text{ }e^{(\widehat{J}(||\xi||_{p}%
)-1)s})\ast1_{E}\\
&  =\left[  \mathcal{F}_{\xi\rightarrow x}^{-1}(e^{(\widehat{J}(||\xi
||_{p})-1)t})\ast\mathcal{F}_{\xi\rightarrow x}^{-1}(e^{(\widehat{J}%
(||\xi||_{p})-1)s})\right]  \ast1_{E},
\end{align*}
since $e^{(\widehat{J}(||\xi||_{p})-1)t}$, $e^{(\widehat{J}(||\xi||_{p}%
)-1)s}\in L_{loc}^{1}$. The Chapman-Kolmogorov equation is exactly
(\ref{Ch_K_1}). Indeed, by using the fact that the convolution of
distributions is associative, we get from (\ref{Ch_K_1}) that%
\begin{align}
p_{t+s}(x,E)  & =\mathcal{F}_{\xi\rightarrow x}^{-1}(e^{(\widehat{J}%
(||\xi||_{p})-1)t})\ast(\mathcal{F}_{\xi\rightarrow x}^{-1}(e^{(\widehat
{J}(||\xi||_{p})-1)s})\ast1_{E})\nonumber\\
& =\mathcal{F}_{\xi\rightarrow x}^{-1}(e^{(\widehat{J}(||\xi||_{p})-1)t})\ast
p_{s}(x,E),\label{Ch_K_2}%
\end{align}
$E\in\mathcal{B}_{0}$. We now recall that the convolution of a distribution
and a test function is a locally constant function, and hence its Fourier
transform, as a distribution, is a function with compact support. By using
this, from (\ref{Ch_K_2}), we have
\begin{align*}
p_{t+s}(x,E)  & =\mathcal{%
\mathcal{F}%
}_{\xi\rightarrow x}^{-1}\left(  e^{(\widehat{J}(||\xi||_{p})-1)t}%
\widehat{p_{s}}(\xi,E)\right) \\
& =\int_{\mathbf{%
\mathbb{Q}
}_{p}^{n}}\chi_{p}\left(  -\xi\cdot x\right)  e^{(\widehat{J}(||\xi
||_{p})-1)t}\widehat{p_{s}}(\xi,E)d^{n}\xi
\end{align*}
in \ $\mathcal{D}^{\prime}(\mathbb{Q}_{p}^{n})$. Let $B_{N}^{n}$ be a ball
containing the support of $\widehat{p_{s}}(\xi,E)$, with $N$ depending on $E $
and $s$, by using Fubini's Theorem,%
\begin{gather}
p_{t+s}(x,E)=\int_{\mathbf{%
\mathbb{Q}
}_{p}^{n}}\chi_{p}\left(  -\xi\cdot x\right)  \left(  e^{(\widehat{J}%
(||\xi||_{p})-1)t}1_{B_{N}^{n}}(\xi)\right)  \widehat{p_{s}}(\xi,E)d^{n}%
\xi\nonumber\\
=\int_{\mathbf{%
\mathbb{Q}
}_{p}^{n}}\chi_{p}\left(  -\xi\cdot x\right)  \left(  e^{(\widehat{J}%
(||\xi||_{p})-1)t}1_{B_{N}^{n}}(\xi)\right)  \left\{  \int_{\mathbf{%
\mathbb{Q}
}_{p}^{n}}\chi_{p}\left(  y\cdot\xi\right)  p_{s}(y,E)d^{n}y\right\}  d^{n}%
\xi\nonumber\\
=\int_{\mathbf{%
\mathbb{Q}
}_{p}^{n}}\left(  \int_{\mathbf{%
\mathbb{Q}
}_{p}^{n}}\chi_{p}\left(  -(x-y)\cdot\xi\right)  e^{(J(||\xi||_{p}%
)-1)t}1_{B_{N}^{n}}(\xi)d^{n}\xi\right)  p_{s}(y,E)d^{n}y\nonumber\\
=\int_{\mathbf{%
\mathbb{Q}
}_{p}^{n}}p_{t}(x-y,\widehat{1}_{B_{N}^{n}})p_{s}(y,E)d^{n}y\text{ for }%
E\in\mathcal{B}_{0}.\label{Ch_K_3}%
\end{gather}
Formula (\ref{Ch_K_3}) between positive distributions extends to a formula
between Borel measures on $\mathbf{%
\mathbb{Q}
}_{p}^{n}$, by the Caratheodory Theorem.
\end{proof}

\begin{remark}
\label{remark6} (i) The transition function $p_{t}(x,\cdot)$ is normal, $i.e.$
$\lim_{t\rightarrow0^{+}}p_{t}(x,%
\mathbb{Q}
_{p}^{n})=1$ for all $x\in%
\mathbb{Q}
_{p}^{n}$. This follows from the fact that $p_{t}(x,%
\mathbb{Q}
_{p}^{n})=1$, see proof of Claim 1.

(ii) From (\ref{Ch_K_1}) we have $\{p_{t}(x,\cdot)\}_{t\geq0}$ is an
convolution semigroup in $\mathcal{D}^{\prime}(%
\mathbb{Q}
_{p}^{n}),$ and moreover $p_{t}(x,\cdot)\rightarrow\delta$ when $t\rightarrow
0^{+}.$

(iii) A function $p(x,E),$ $x\in%
\mathbb{Q}
_{p}^{n},$ $E\in\mathcal{B}(%
\mathbb{Q}
_{p}^{n}),$ is called a sub-Markovian transition function on $(%
\mathbb{Q}
_{p}^{n},\mathcal{B}(%
\mathbb{Q}
_{p}^{n})),$ if it satisfies $(1)$ for every $x\in%
\mathbb{Q}
_{p}^{n},$ $p(x,\cdot)$ is a measure on $\mathcal{B}(%
\mathbb{Q}
_{p}^{n})$ such that $p(x,%
\mathbb{Q}
_{p}^{n})\leq1;$ $(2)$ for every $E\in\mathcal{B}(%
\mathbb{Q}
_{p}^{n}),$ $p(\cdot,E)$ is a Borel measurable function. Therefore,
$p_{t}(x,E)$ for $t\geq0,$ $x\in%
\mathbb{Q}
_{p}^{n}$, $E\in\mathcal{B}(%
\mathbb{Q}
_{p}^{n})$ is a sub-Markov transition function on $%
\mathbb{Q}
_{p}^{n}.$
\end{remark}

Let $C_{b}(%
\mathbb{Q}
_{p}^{n})$ be the space of real-valued, bounded, and continuous functions on $%
\mathbb{Q}
_{p}^{n}$. This is a Banach space with the norm $\left\Vert f\right\Vert
_{\infty}=\sup_{x\in%
\mathbb{Q}
_{p}^{n}}|f(x)|$. We say that a function $f\in C_{b}(%
\mathbb{Q}
_{p}^{n})$ converges to zero as $x\rightarrow\infty$ if, for each $\epsilon
>0$, there exists a compact subset $E\subset%
\mathbb{Q}
_{p}^{n}$ such that $\left\vert f\left(  x\right)  \right\vert <\epsilon$ for
all $x\in%
\mathbb{Q}
_{p}^{n}\smallsetminus E$. In such case we write $\lim_{x\rightarrow\infty
}f(x)=0$. We set%
\[
C_{0}(%
\mathbb{Q}
_{p}^{n}):=\{f\in C_{b}(%
\mathbb{Q}
_{p}^{n});\text{ }\lim_{x\rightarrow\infty}f(x)=0\}.
\]
The space $C_{0}(%
\mathbb{Q}
_{p}^{n})$ is a closed subspace of $C_{b}(%
\mathbb{Q}
_{p}^{n})$, and thus it is a Banach space.

\begin{definition}
Given a Markov transition function $p_{t}(x,\cdot)$, we attach to it the
following operator:%
\[
T_{t}f(x):=\left\{
\begin{array}
[c]{lll}%
\int_{%
\mathbb{Q}
_{p}^{n}}p_{t}(x,d^{n}y)f(y) & \text{if} & t>0\\
&  & \\
f & \text{if} & t=0.
\end{array}
\right.
\]
We say that $p_{t}(x,\cdot)$ is a $C_{0}$-function if the space $C_{0}(%
\mathbb{Q}
_{p}^{n})$ is an invariant subspace for the operators $T_{t},$ $t\geq0$, i.e.%
\[
f\in C_{0}(%
\mathbb{Q}
_{p}^{n})\longrightarrow T_{t}f\in C_{0}(%
\mathbb{Q}
_{p}^{n}).
\]

\end{definition}

\begin{lemma}
\label{lemma7} $p_{t}(x,\cdot)$ is a $C_{0}$-function. Furthermore,
$T_{t}:C_{0}(%
\mathbb{Q}
_{p}^{n})\rightarrow C_{0}(%
\mathbb{Q}
_{p}^{n})$ is a bounded linear operator.
\end{lemma}

\begin{proof}
The result follows from the fact that $\mathcal{D}(\mathbb{Q}_{p}^{n})$ is
dense in $C_{0}(%
\mathbb{Q}
_{p}^{n})$, see e.g. \cite[Proposition 1.3]{Taibleson}, by the following Claim:

\textbf{Claim.} $T_{t}:\left(  \mathcal{D}(\mathbb{Q}_{p}^{n}),\left\Vert
\cdot\right\Vert _{\infty}\right)  \rightarrow C_{0}(%
\mathbb{Q}
_{p}^{n})$ is a bounded operator.

The proof of the Claim is as follows. Take $f\in\mathcal{D}(\mathbb{Q}_{p}%
^{n})$ and $t>0$, then%
\begin{align*}
|T_{t}f(x)|  & =|\int_{%
\mathbb{Q}
_{p}^{n}}p_{t}(x,d^{n}y)f(y)|\leq\int_{%
\mathbb{Q}
_{p}^{n}}p_{t}(x,d^{n}y)\left\vert f(y)\right\vert \\
& =\int_{%
\mathbb{Q}
_{p}^{n}\smallsetminus\left\{  0\right\}  }\widetilde{Z}_{t}(y)\left\vert
f(x-y)\right\vert d^{n}y\leq\left\Vert f\right\Vert _{\infty}\int_{%
\mathbb{Q}
_{p}^{n}\smallsetminus\left\{  0\right\}  }\widetilde{Z}_{t}(y)d^{n}y\\
& =\left\Vert f\right\Vert _{\infty}p_{t}(0,%
\mathbb{Q}
_{p}^{n})=\left\Vert f\right\Vert _{\infty},
\end{align*}
cf. \textbf{Claim 1} in the proof of Proposition \ref{prop2}, this shows that
$T_{t}$ is a linear bounded operator from $\left(  \mathcal{D}(\mathbb{Q}%
_{p}^{n}),\left\Vert \cdot\right\Vert _{\infty}\right)  $ into $L^{\infty}(%
\mathbb{Q}
_{p}^{n})$. We now show that $\lim_{x\rightarrow\infty}T_{t}f(x)=0$. Take
$f\in\mathcal{D}(\mathbb{Q}_{p}^{n})$, with supp$f=E$ and $t>0$, then%
\begin{align*}
|T_{t}f(x)| &  =|\int_{%
\mathbb{Q}
_{p}^{n}}p_{t}(x,d^{n}y)f(y)|\leq\left\Vert f\right\Vert _{\infty}\int
_{E}\widetilde{Z}_{t}(x-y)d^{n}y\\
&  \leq Ct\left\Vert f\right\Vert _{\infty}\int_{E}||x-y||_{p}^{-n}%
d^{n}y=Ct\left\Vert f\right\Vert _{\infty}||x||_{p}^{-n}vol(E),
\end{align*}
for $||x||_{p}$ big enough, cf. Proposition \ref{prop3}-(i). Finally, we show
that $\lim_{x\rightarrow x_{0}}T_{t}f\left(  x\right)  =T_{t}f\left(
x_{0}\right)  $ for $t>0$. This fact follows by using the Dominated
Convergence Theorem, since%
\[
T_{t}f\left(  x\right)  =\int_{%
\mathbb{Q}
_{p}^{n}\smallsetminus\left\{  0\right\}  }\widetilde{Z}_{t}(y)f\left(
x-y\right)  d^{n}y
\]
and $\left\vert \widetilde{Z}_{t}(y)f\left(  x-y\right)  \right\vert
\leq\left\Vert f\right\Vert _{\infty}\widetilde{Z}_{t}(y)$ with $\int_{%
\mathbb{Q}
_{p}^{n}\smallsetminus\left\{  0\right\}  }\widetilde{Z}_{t}(y)d^{n}y=p_{t}(0,%
\mathbb{Q}
_{p}^{n})=1$, cf. \textbf{Claim 1} in the proof of Proposition \ref{prop2}.
\end{proof}

\begin{remark}
\label{remark7} $(i)$ We recall some results on Hunt, L\'{e}vy and Markov
processes that we need to establish the main theorem of this section. All our
processes have state space $(%
\mathbb{Q}
_{p}^{n},\mathcal{B}\left(
\mathbb{Q}
_{p}^{n}\right)  ).$ A Hunt process is a Markov standard process which is
quasi-left continuous on $[0,\infty),$ see \cite[Definition 9.2 and the
accompanying remarks]{Blumenthal-Getoor}.

$(ii)$ Let $X=\{X_{t}\}_{t\geq0}$ be a Hunt process with state space $%
\mathbb{Q}
_{p}^{n}$ and adjoined terminal state $\partial$ is a L\'{e}vy process on $%
\mathbb{Q}
_{p}^{n}$ if $(1)$ $P^{x}(X_{t}\in E)=P^{0}(X_{t}+x\in E)$ for $t\geq0,$ $0, $
$x\in%
\mathbb{Q}
_{p}^{n}$ and $E$ a Borel subset of $%
\mathbb{Q}
_{p}^{n};$ and $(2)$ $P^{0}(X_{t}\in%
\mathbb{Q}
_{p}^{n})=1$ for $t\geq0.$ Here $p_{t}(x,E)=P^{x}(X_{t}\in E).$

$(iii)$ The family of Borel probability measures $\{\mu_{t},$ $t\geq0\}$ given
by
\begin{equation}
\mu_{t}(E)=P^{0}(X_{t}\in E)\label{obs mu_t}%
\end{equation}
is a convolution semigroup such that
\begin{equation}
\mu_{t}\rightarrow\delta\text{ as }t\rightarrow0^{+},\label{mu_t tends}%
\end{equation}
where $\delta$ denotes the Dirac distribution. Conversely, it can be shown
that for any convolution semigroup $\{\mu_{t},$ $t\geq0\}$ satisfying
(\ref{mu_t tends}), it is possible to construct a L\'{e}vy process $X_{t}$
with state space $%
\mathbb{Q}
_{p}^{n}$ such that (\ref{obs mu_t}) is satisfied, see \cite[Section 2]{Evans}
and \cite[Exercise I-9-14]{Blumenthal-Getoor}.
\end{remark}

\begin{theorem}
\label{Theorem2} There exists a L\'{e}vy process $\mathfrak{X}\left(
t,\omega\right)  $, with state space $(%
\mathbb{Q}
_{p}^{n},\mathcal{B}\left(
\mathbb{Q}
_{p}^{n}\right)  )$ and transition function $p_{t}(x,\cdot)$.
\end{theorem}

\begin{proof}
We first show that there exists a Hunt process $\mathfrak{X}\left(
t,\omega\right)  $ with state space $(%
\mathbb{Q}
_{p}^{n},\mathcal{B}\left(
\mathbb{Q}
_{p}^{n}\right)  )$ and transition function $p_{t}(x,\cdot).$ This result
follow from \cite[Theorem 9.4]{Blumenthal-Getoor} by Proposition \ref{prop2},
Lemma \ref{lemma7} and Remark \ref{remark6}-(i), (iii). On the other hand,
from Remarks \ref{remark7} and \ref{remark6}-(ii), it follows that the Hunt
process constructed is a L\'{e}vy process.
\end{proof}

\section{First Passage Time Problem}

Consider the following Cauchy problem:%
\begin{equation}
\left\{
\begin{array}
[c]{ll}%
\frac{\partial u}{\partial t}(x,t)=Au(x,t), & t\in\left[  0,\infty\right)
,\text{\ }x\in%
\mathbb{Q}
_{p}^{n}\\
& \\
u(x,0)=\Omega(||x||_{p}). &
\end{array}
\right. \label{Cauchy_problem_2}%
\end{equation}
By Proposition \ref{lemma_prop1},
\begin{equation}
u(x,t)=Z_{t}(x)\ast\Omega(||x||_{p}),\label{u(x,t)}%
\end{equation}
is a classical solution of (\ref{Cauchy_problem_2}). We now define%
\[
q_{t}\left(  x,E\right)  =\left\{
\begin{array}
[c]{ll}%
\left(  u(\cdot,t)\ast1_{E}\right)  \left(  x\right)  & \text{for }t>0\text{
and }E\in\mathcal{B}\left(
\mathbb{Q}
_{p}^{n}\right) \\
& \\
1_{E}\left(  x\right)  & \text{for }t=0\text{ and }E\in\mathcal{B}\left(
\mathbb{Q}
_{p}^{n}\right)  .
\end{array}
\right.
\]
Since
\[%
\begin{array}
[c]{ccc}%
\mathcal{D}\left(
\mathbb{Q}
_{p}^{n}\right)  & \rightarrow & \mathcal{%
\mathbb{C}
}\\
&  & \\
\phi\left(  x\right)  & \rightarrow & \phi\left(  x\right)  \ast
\Omega(||x||_{p})
\end{array}
\]
is linear continuous mapping, by the arguments given in the proof of
Proposition \ref{prop2}, $q_{t}\left(  x,E\right)  $ is the transition
function of a Markov process $\mathfrak{J}(t,\omega)$. Set $\Upsilon$ to be
the space of all paths $\mathfrak{J}(t,\omega).$ Then there exists a
probability space $\left(  \Upsilon,\mathcal{F},P\right)  ,$ where $P$ is a
probability measure on $\Upsilon$ and $\mathfrak{J}(t,\cdot):\left(
\Upsilon,\mathcal{F},P\right)  \rightarrow\left(
\mathbb{Q}
_{p}^{n},\mathcal{B}\left(
\mathbb{Q}
_{p}^{n}\right)  ,d^{n}x\right)  $ is random variable for each $t\geq0$. The
construction of this probability space follows from classical arguments, see,
e.g., \cite{Edward N}. We notice that%
\begin{align*}
P(\{\omega & \in\Upsilon:\text{ }\mathfrak{J}(0,\omega)\in%
\mathbb{Z}
_{p}^{n}\})=q_{0}\left(  0,\mathbb{Z}_{p}^{n}\right)  =\Omega(||x||_{p}%
)\ast\Omega\left(  ||x||_{p}\right)  \mid_{x=0}\\
& =\Omega\left(  ||x||_{p}\right)  \mid_{x=0}=1\text{.}%
\end{align*}

In this section we study the following random variable.

\begin{definition}
The random variable $\tau_{%
\mathbb{Z}
_{p}^{n}}(\omega):$ $\Upsilon\longrightarrow%
\mathbb{R}
_{+}\cup\{+\infty\}$ defined by%
\[
\inf\{t>0;\mathfrak{J(}t,\omega\mathfrak{)}\in%
\mathbb{Z}
_{p}^{n}\text{ and there exists }t^{\prime}\text{ such that }0<t^{\prime
}<t\text{ and }\mathfrak{J(}t^{\prime},\omega\mathfrak{)\notin}%
\mathbb{Z}
_{p}^{n}\}
\]
is called the first passage time of a path of the random process
$\mathfrak{J(}t,\omega\mathfrak{)}$ entering the domain $%
\mathbb{Z}
_{p}^{n}$.
\end{definition}

\begin{remark}
\label{remark8} We notice that the condition
\begin{equation}
P(\{\omega\in\Upsilon:\tau_{%
\mathbb{Z}
_{p}^{n}}(\omega)\mathfrak{<}\infty\})=1\label{Probability}%
\end{equation}
means that every path of $\mathfrak{J(}t,\omega\mathfrak{)}$ is sure to return
to $%
\mathbb{Z}
_{p}^{n}$. If (\ref{Probability}) does not hold, then there exist paths of
$\mathfrak{J(}t,\omega\mathfrak{)}$ that abandon $%
\mathbb{Z}
_{p}^{n}$ and never go back.
\end{remark}

\begin{lemma}
\label{lemma8}The function $u(x,t)=Z_{t}(x)\ast\Omega(||x||_{p})$, $t\geq0 $,
is pointwise differentiable in $t$ and its derivative is given by the formula%
\[
\frac{\partial u}{\partial t}(x,t)=%
{\displaystyle\int\limits_{\mathbb{Z}_{p}^{n}}}
\chi_{p}(-x.\xi)e^{(\widehat{J}(\left\Vert \xi\right\Vert _{p})-1)t}\left[
\widehat{J}(\left\Vert \xi\right\Vert _{p})-1\right]  d^{n}\xi\text{, for
}t\geq0.
\]

\end{lemma}

\begin{proof}
The formula is obtained by applying the Dominated Convergence Theorem.
\end{proof}

\begin{lemma}
\label{lemma9}The probability density function for a path of $\mathfrak{J}%
(t,\omega)$ to enter into $%
\mathbb{Z}
_{p}^{n}$ at the instant of time $t$, with the condition that $\mathfrak{J}%
(0,\omega)\in$ $%
\mathbb{Z}
_{p}^{n}$ is given by%
\begin{equation}
g(t)=%
{\displaystyle\int\limits_{\mathbf{\mathbb{Q} }_{p}^{n}\backslash\mathbb{Z}
_{p}^{n}}}
J(||y||_{p})u(y,t)d^{n}y.\label{g(t)}%
\end{equation}

\end{lemma}

\begin{proof}
The survival probability, by definition%
\[
S(t):=S_{%
\mathbb{Z}
_{p}^{n}}(t)=%
{\displaystyle\int\limits_{\mathbb{Z} _{p}^{n}}}
u(x,t)d^{n}x,
\]
is the probability that a path of $\mathfrak{J(}t,\omega\mathfrak{)}$ remains
in $%
\mathbb{Z}
_{p}^{n}$ at the time $t.$ Because there are no external or internal sources,%
\begin{align*}
S^{\prime}(t) &  =\left.
\begin{array}
[c]{l}%
\text{Probability that a path of }\mathfrak{J}(t,\omega)\\
\text{goes back to }%
\mathbb{Z}
_{p}^{n}\text{ at the time }t
\end{array}
\right.  -\left.
\begin{array}
[c]{l}%
\text{Probability that a path of }\mathfrak{J}(t,\omega)\\
\text{exists }%
\mathbb{Z}
_{p}^{n}\text{ at the time }t
\end{array}
\right. \\
& \\
&  =g(t)-CS(t)\text{, with }0<C\leq1.
\end{align*}
By using Lemma \ref{lemma8},
\begin{gather*}
S^{\prime}(t)=%
{\displaystyle\int\limits_{\mathbb{Z} _{p}^{n}}}
\frac{\partial}{\partial t}u(x,t)d^{n}x=%
{\displaystyle\int\limits_{\mathbb{Z} _{p}^{n}}}
\left\{
{\displaystyle\int\limits_{\mathbf{\mathbb{Q} }_{p}^{n}}}
J(||y||_{p})u(x-y,t)d^{n}y-u(x,t)\right\}  d^{n}x\\
=%
{\displaystyle\int\limits_{\mathbb{Z} _{p}^{n}}}
{\displaystyle\int\limits_{\mathbf{\mathbb{Q} }_{p}^{n}}}
J(||y||_{p})\left\{  u(x-y,t)-u(x,t)\right\}  d^{n}yd^{n}x\\
=%
{\displaystyle\int\limits_{\mathbb{Z} _{p}^{n}}}
{\displaystyle\int\limits_{\mathbb{Z} _{p}^{n}}}
J(||y||_{p})\left\{  u(x-y,t)-u(x,t)\right\}  d^{n}yd^{n}x\\
+%
{\displaystyle\int\limits_{\mathbb{Z} _{p}^{n}}}
{\displaystyle\int\limits_{\mathbf{\mathbb{Q} }_{p}^{n}\backslash\mathbb{Z}
_{p}^{n}}}
J(||y||_{p})\left\{  u(x-y,t)-u(x,t)\right\}  d^{n}yd^{n}x.
\end{gather*}
By Proposition \ref{lemma_prop1}, for $x,y\in%
\mathbb{Z}
_{p}^{n}$,
\begin{align*}
u(x-y,t)  & =\int_{\mathbb{Z}_{p}^{n}}\chi_{p}\left(  -(x-y)\cdot\xi\right)
e^{(\widehat{J}(||\xi||_{p})-1)t}d^{n}\xi=\int_{%
\mathbb{Z}
_{p}^{n}}e^{(\widehat{J}(||\xi||_{p})-1)t}d^{n}\xi\\
& =u(x,t)\text{, }%
\end{align*}
i.e. $u(x-y,t)-u(x,t)\equiv0$ for $x,y\in%
\mathbb{Z}
_{p}^{n}$, consequently,%
\begin{align*}
S^{\prime}(t)  & =%
{\displaystyle\int\limits_{\mathbb{Z} _{p}^{n}}}
{\displaystyle\int\limits_{\mathbf{\mathbb{Q} }_{p}^{n}\backslash\mathbb{Z}
_{p}^{n}}}
J(||y||_{p})(u(x-y,t)-u(x,t))d^{n}yd^{n}x\\
& =%
{\displaystyle\int\limits_{\mathbb{Z} _{p}^{n}}}
{\displaystyle\int\limits_{\mathbf{\mathbb{Q} }_{p}^{n}\backslash\mathbb{Z}
_{p}^{n}}}
J(||y||_{p})u(x-y,t)d^{n}yd^{n}x-%
{\displaystyle\int\limits_{\mathbf{\mathbb{Q} }_{p}^{n}\backslash\mathbb{Z}
_{p}^{n}}}
J(||y||_{p})d^{n}y%
{\displaystyle\int\limits_{\mathbb{Z} _{p}^{n}}}
u(x,t)d^{n}x\\
& =%
{\displaystyle\int\limits_{\mathbb{Z} _{p}^{n}}}
{\displaystyle\int\limits_{\mathbf{\mathbb{Q} }_{p}^{n}\backslash\mathbb{Z}
_{p}^{n}}}
J(||y||_{p})u(x-y,t)d^{n}yd^{n}x-CS(t),
\end{align*}
with $C:=\int_{\mathbf{%
\mathbb{Q}
}_{p}^{n}\backslash%
\mathbb{Z}
_{p}^{n}}J(||y||_{p})d^{n}y\leq1$, since $J$ is of exponential type and
$\int_{\mathbf{%
\mathbb{Q}
}_{p}^{n}}J(||y||_{p})d^{n}y=1$. We notice that if $x\in%
\mathbb{Z}
_{p}^{n}$ and $y\in\mathbf{%
\mathbb{Q}
}_{p}^{n}\backslash%
\mathbb{Z}
_{p}^{n}$, then
\begin{gather*}
u(x-y,t)=\int_{%
\mathbb{Z}
_{p}^{n}}\chi_{p}\left(  -x\cdot\xi\right)  \chi_{p}\left(  y\cdot\xi\right)
e^{(\widehat{J}(||\xi||_{p})-1)t}d^{n}\xi\\
=\int_{%
\mathbb{Z}
_{p}^{n}}\chi_{p}\left(  y\cdot\xi\right)  e^{(\widehat{J}(||\xi||_{p}%
)-1)t}d^{n}\xi=\int_{%
\mathbb{Z}
_{p}^{n}}\chi_{p}\left(  -y\cdot\xi\right)  e^{(\widehat{J}(||\xi||_{p}%
)-1)t}d^{n}\xi=u(y,t),
\end{gather*}
and consequently%
\begin{align*}
S^{\prime}(t)  & =%
{\displaystyle\int\limits_{\mathbf{\mathbb{Q} }_{p}^{n}\backslash\mathbb{Z}
_{p}^{n}}}
J(||y||_{p})u(y,t)d^{n}y-CS(t)\\
& \\
& =g(t)-CS(t)\text{\ with }0<C\leq1.
\end{align*}

\end{proof}

\begin{proposition}
\label{prop4}The probability density function $f(t)$ of the random variable
$\tau_{%
\mathbb{Z}
_{p}^{n}}(\omega)$ satisfies the non-homogeneous Volterra equation of second
kind%
\begin{equation}
g(t)=%
{\displaystyle\int\limits_{0}^{\infty}}
g(t-\tau)f(\tau)d\tau+f(t).\label{Volterra_equ}%
\end{equation}

\end{proposition}

\begin{proof}
The result follow from Lemma \ref{lemma9} by using the argument given in the
proof of Theorem 1 in \cite{Av-2}.
\end{proof}

\begin{lemma}
\label{prop5}The Laplace transform $G(s)$ of $g(t)$ is given by%
\begin{equation}
G(s)=%
{\displaystyle\int\limits_{\mathbf{\mathbb{Q} }_{p}^{n}\backslash\mathbb{Z}
_{p}^{n}}}
J(||y||_{p})%
{\displaystyle\int\limits_{\mathbb{Z} _{p}^{n}}}
\frac{\chi_{p}\left(  -y\cdot\xi\right)  }{s+(1-\widehat{J}(\left\Vert
\xi\right\Vert _{p}))}d^{n}\xi d^{n}y\text{, for }\operatorname{Re}%
(s)>0.\label{Laplace_T}%
\end{equation}

\end{lemma}

\begin{proof}
We first note that $e^{-st}J(||y||_{p})e^{(\widehat{J}(||\xi||_{p})-1)t}%
\Omega(||\xi||_{p})\in L^{1}((0,\infty)\times\mathbf{%
\mathbb{Q}
}_{p}^{n}\backslash%
\mathbb{Z}
_{p}^{n}\times\mathbf{%
\mathbb{Q}
}_{p}^{n},dtd^{n}\xi d^{n}y)$ for $s\in%
\mathbb{C}
$ with $\operatorname{Re}(s)>0$. The announced formula follows now from
(\ref{g(t)}) and (\ref{u(x,t)}) by using Fubini's Theorem.
\end{proof}

\begin{theorem}
\label{Theorem3} If $-n<\gamma<0$, then $\ P(\{\omega\in\Upsilon:\tau_{%
\mathbb{Z}
_{p}^{n}}(\omega)\mathfrak{<}\infty\})=1$.
\end{theorem}

\begin{proof}
By applying the Laplace transform to (\ref{Volterra_equ}), we have%
\[
F(s)=\frac{G(s)}{1+G(s)}=1-\frac{1}{1+G(s)},
\]
where $F(s)$ and $G(s)$ are the Laplace transforms of $f$ and $g$,
respectively. We understand $F(0)=\lim_{s\rightarrow0}F(s)$ and $G(0)=\lim
_{s\rightarrow0}G(s)$. From $F(0)=\int_{0}^{\infty}f(t)dt=\frac{G(0)}{1+G(0)}%
$, it follows that if $G(0)=\infty$, then $\mathfrak{J(}t,\omega\mathfrak{)}$
is recurrent. Since $G(0)=\int_{0}^{\infty}g(t)dt$ is either a positive number
or infinity, it is sufficient to show that $\lim_{s\rightarrow0}G(s)=\infty$
for $s\in\mathbb{R}_{+}$. For $y\in\mathbf{%
\mathbb{Q}
}_{p}^{n}\backslash%
\mathbb{Z}
_{p}^{n}$ with $||y||_{p}=p^{i},$ $i\in%
\mathbb{N}
\backslash\{0\}$, we have%
\begin{align*}%
{\displaystyle\int\limits_{\mathbb{Z} _{p}^{n}}}
\frac{\chi_{p}\left(  -y\cdot\xi\right)  }{s+1-\widehat{J}(\left\Vert
\xi\right\Vert _{p})}d^{n}\xi & =%
{\displaystyle\sum\limits_{j=0}^{\infty}}
\frac{1}{s+1-\widehat{J}(p^{-j})}%
{\displaystyle\int\limits_{||\xi||_{p}=p^{-j}}}
\chi_{p}\left(  -y\cdot\xi\right)  d^{n}\xi\\
& =%
{\displaystyle\sum\limits_{j=i}^{\infty}}
\frac{p^{-nj}(1-p^{-n})}{s+1-\widehat{J}(p^{-j})}-\frac{p^{-ni}}%
{s+1-\widehat{J}(p^{1-i})},
\end{align*}
and by (\ref{Laplace_T}),%
\begin{align*}
G(s)  & =%
{\displaystyle\sum\limits_{i=1}^{\infty}}
J(p^{i})%
{\displaystyle\int\limits_{\mathbf{||}y\mathbf{||}_{p}=p^{i}}}
\left[
{\displaystyle\sum\limits_{j=i}^{\infty}}
\frac{p^{-nj}(1-p^{-n})}{s+1-\widehat{J}(p^{-j})}-\frac{p^{-ni}}%
{s+1-\widehat{J}(p^{1-i})}\right]  d^{n}y\\
& =%
{\displaystyle\sum\limits_{i=1}^{\infty}}
J(p^{i})\left[
{\displaystyle\sum\limits_{j=i}^{\infty}}
\frac{p^{-nj}(1-p^{-n})}{s+1-\widehat{J}(p^{-j})}-\frac{p^{-ni}}%
{s+1-\widehat{J}(p^{1-i})}\right]  p^{ni}(1-p^{-n})\\
& =(1-p^{-n})%
{\displaystyle\sum\limits_{i=1}^{\infty}}
J(p^{i})\left[  (1-p^{-n})%
{\displaystyle\sum\limits_{j=i}^{\infty}}
\frac{p^{n(i-j)}}{s+1-\widehat{J}(p^{-j})}-\frac{1}{s+1-\widehat{J}(p^{1-i}%
)}\right]  .
\end{align*}
Now, since\ $\lim_{j\rightarrow\infty}1-\widehat{J}(p^{-j})=0$, given any
$s>0$, there exists $j_{0}\left(  s\right)  \in\mathbb{N}$ such that
$1-\widehat{J}(p^{-j})<s$ for $j>$ $j_{0}\left(  s\right)  $. In addition,
$s\rightarrow0^{+}$ implies that $j_{0}\left(  s\right)  \rightarrow\infty$.
By using these observations, we have%
\begin{align*}
G(s)  & \geq(1-p^{-n})J(p)\left[  (1-p^{-n})%
{\displaystyle\sum\limits_{j=1}^{\infty}}
\frac{p^{n(1-j)}}{s+1-\widehat{J}(p^{-j})}-\frac{1}{s+1-\widehat{J}(1)}\right]
\\
& \geq(1-p^{-n})J(p)\left[  \frac{p^{n}\left(  1-p^{-n}\right)  }{2}%
{\displaystyle\sum\limits_{j=1}^{j_{0}\left(  s\right)  }}
\frac{p^{-nj}}{1-\widehat{J}(p^{-j})}-\frac{1}{s+1-\widehat{J}(1)}\right]  ,
\end{align*}
notice that by Remark \ref{Nota_support}, $1-\widehat{J}(1)>0$. Therefore,
\begin{align*}
\lim_{s\rightarrow0^{+}}G(s)  & \geq(1-p^{-n})J(p)\left[  \frac{p^{n}\left(
1-p^{-n}\right)  }{2}%
{\displaystyle\sum\limits_{j=0}^{\infty}}
\frac{p^{-nj}}{1-\widehat{J}(p^{-j})}-\frac{1+\frac{p^{n}\left(
1-p^{-n}\right)  }{2}}{1-\widehat{J}(1)}\right] \\
& =(1-p^{-n})J(p)\left[  \frac{p^{n}}{2}%
{\displaystyle\int\limits_{\mathbb{Z} _{p}^{n}}}
\frac{d^{n}\xi}{1-\widehat{J}(\left\Vert \xi\right\Vert _{p})}-\frac
{1+\frac{p^{n}\left(  1-p^{-n}\right)  }{2}}{1-\widehat{J}(1)}\right]
=\infty,
\end{align*}
cf. Lemma \ref{lemma2A}.
\end{proof}


\bigskip

\begin{thebibliography}{99}                                                                                               %
\bibitem {Alberio et al}Albeverio S., Khrennikov A. Yu., Shelkovich V. M.,
Theory of $p$-adic distributions: linear and nonlinear models. London
Mathematical Society Lecture Note Series, 370. Cambridge University Press,
Cambridge, 2010.

\bibitem {A-K3}Albeverio, Sergio, Karwowski Witold, Jump processes on leaves
of multibranching trees, J. Math. Phys. 49 (2008), no. 9, 093503, 20 pp.

\bibitem {Andreu-Vaillo et al}Andreu-Vaillo Fuensanta, Maz\'{o}n Jos\'{e} M.,
Rossi Julio D., Toledo-Melero J. Juli\'{a}n, Nonlocal diffusion problems.
Mathematical Surveys and Monographs, 165. American Mathematical Society,
Providence, RI; Real Sociedad Matem\'{a}tica Espa\~{n}ola, Madrid, 2010.

\bibitem {Av-2}Avetisov V. A., Bikulov A. Kh., Zubarev, A. P., First passage
time distribution and the number of returns for ultrametric random walks, J.
Phys. A 42 (2009), no. 8, 085003, 18 pp.

\bibitem {Av-4}Avetisov V. A., Bikulov A. Kh., Osipov V. A., $p$-adic
description of characteristic relaxation in complex systems, J. Phys. A 36
(2003), no. 15, 4239--4246.

\bibitem {Av-5}Avetisov V. A., Bikulov A. H., Kozyrev S. V., Osipov V. A., $p
$-adic models of ultrametric diffusion constrained by hierarchical energy
landscapes, J. Phys. A 35 (2002), no. 2, 177--189.

\bibitem {Av-7}Avetisov V. A., Bikulov A. Kh., Kozyrev S. V., Description of
logarithmic relaxation by a model of a hierarchical random walk. (Russian)
Dokl. Akad. Nauk 368 (1999), no. 2, 164--167.

\bibitem {Becker  et al}Becker O. M., Karplus M., The topology of
multidimensional protein energy surfaces: theory and application to peptide
structure and kinetics, J. Chem.Phys. 106, 1495--1517 (1997).

\bibitem {Bendikov}Bendikov A. D., Grigor$^{,}$yan A. A., Pittet C., Woess W.,
Isotropic Markov semigroups on ultra-metric spaces. (Russian) Uspekhi Mat.
Nauk 69 (2014), no 4(418), 3-102; translation in Russian Math. Surveys 69
(2014), no. 4,589-680.

\bibitem {Bendikov2}Bendikov A., Heat kernels for isotropic-like Markov
generators on ultrametric spaces: a Survey, $p$-Adic Numbers Ultrametric Anal.
Appl. 10 (2018), no. 1, 1--11. 

\bibitem {Berg-Gunnar}Berg Christian, Forst Gunnar, Potential theory on
locally compact abelian groups. Springer-Verlag, New York-Heidelberg, 1975.

\bibitem {Blumenthal-Getoor}Blumenthal R. M., Getoor R. K., Markov processes
and potential Theory. Academic Press, New York and London, 1968.

\bibitem {Casas-Zuniga}Casas-S\'{a}nchez O. F., Z\'{u}\~{n}iga-Galindo W. A.,
$p$-adic elliptic quadratic forms, parabolic-type pseudodifferential equations
with variable coefficients and Markov processes, $p$-Adic Numbers Ultrametric
Anal. Appl. 6 (2014), no. 1, 1--20.

\bibitem {C-H}Cazenave Thierry, Haraux Alain, An introduction to semilinear
evolution equations. Oxford University Press, 1998.

\bibitem {Ch-Z-2}Chac\'{o}n-Cortes L. F., Z\'{u}\~{n}iga-Galindo W. A.,
Non-local operators, non-Archimedean parabolic-type equations with variable
coefficients and Markov processes, Publ. Res. Inst. Math. Sci. 51 (2015), no.
2, 289--317.

\bibitem {Ch-Z-1}Chac\'{o}n-Cortes L. F., Z\'{u}\~{n}iga-Galindo W. A.,
Nonlocal operators, parabolic-type equations, and ultrametric random walks. J.
Math. Phys. 54 (2013), no. 11, 113503, 17 pp. Erratum 55 (2014), no. 10,
109901, 1 pp.

\bibitem {Chen-Kumagi}Chen Zhen-Qing, Kumagai Takashi, Heat kernel estimates
for jump processes of mixed types on metric measure spaces, Probab. Theory
Related Fields 140 (2008), no. 1-2, 277--317.

\bibitem {Dra-Kh-K-V}Dragovich B., Khrennikov A. Yu., Kozyrev S. V., Volovich,
I. V., On $p$-adic mathematical physics, p-Adic Numbers Ultrametric Anal.
Appl. 1 (2009), no. 1, 1--17.

\bibitem {Dyn}Dynkin E. B., Markov processes. Vol. I. Springer-Verlag, 1965.

\bibitem {Edward N}Edward N, \textquotedblleft Feynman integrals and the
Schr\"{o}dinger equation\textquotedblright, J. Math. Phys. 5, 332-343 (1964).

\bibitem {E-N}Engel K.-J., Nagel. R., One-Parameter Semigroups for Linear
Evolution Equations. Springer-Verlag, 2000.

\bibitem {Evans}Evans Steven N., Local Properties of L\'{e}vy Processes on a
Totally Disconnected Group, Journal of Theoretical Probability, (1989), Vol.
2, no. 2, 209--259.

\bibitem {Fraunfelder et al}Frauenfelder H, Chan S. S., Chan W. S. (eds), The
Physics of Proteins. Springer-Verlag, 2010.

\bibitem {Fraunfelder et al 2}Frauenfelder H., Sligar S.G., Wolynes P.G., The
energy landscape and motions of proteins, Science 254, 1598--1603 (1991).

\bibitem {Fraunfelder et al 3}Frauenfelder H., McMahon B. H., Fenimore P. W.,
Myoglobin: the hydrogen atom of biology and paradigm of complexity, PNAS 100
(15), 8615--8617 (2003).

\bibitem {Halmos}Halmos Paul R., Measure Theory. Van Nostrand Company, 1950.

\bibitem {Hoffmann}Hoffmann K. H., Sibani P., Diffusion in Hierarchies, Phys.
Rev. A 38, 4261--4270 (1988).

\bibitem {Ka}Karwowski W., Diffusion processes with ultrametric jumps, Rep.
Math. Phys. 60 (2007), no. 2, 221--235.

\bibitem {Kigami}Kigami Jun, Transitions on a noncompact Cantor set and random
walks on its defining tree, Ann. Inst. Henri Poincar\'{e} Probab. Stat. 49
(2013), no. 4, 1090--1129.

\bibitem {KKZuniga}Khrennikov Andrei, Kozyrev Sergei, Z{\'{u}}{\~{n}}%
iga-Galindo W. A., Ultrametric Equations and its Applications. Encyclopedia of
Mathematics and its Applications (168). Cambridge University Press, 2018.

\bibitem {Koch}Kochubei Anatoly N., Pseudo-differential equations and
stochastics over non-Archimedean fields. Marcel Dekker, Inc., New York, 2001.

\bibitem {Kozyrev  SV}Kozyrev S. V., Methods and Applications of Ultrametric
and p-Adic Analysis: From Wavelet Theory to Biophysics, Sovrem. Probl. Mat.,
12, Steklov Math. Inst., RAS, Moscow, 2008, 3--168.

\bibitem {M-P-V}M\'{e}zard Marc, Parisi Giorgio, Virasoro Miguel Angel, Spin
glass theory and beyond. World Scientific, 1987.

\bibitem {Ogielski  et al}Ogielski A. T., Stein D. L., Dynamics on Ultrametric
Spaces, Phys. Rev. Lett. 55 (15), 1634--1637 (1985).

\bibitem {R-T-V}Rammal R., Toulouse G., Virasoro M. A., Ultrametricity for
physicists, Rev. Modern Phys. 58 (1986), no. 3, 765--788.

\bibitem {R-Zu}Rodr\'{\i}guez-Vega J. J., Z\'{u}\~{n}iga-Galindo W. A.,
Taibleson operators, $p$-adic parabolic equations and ultrametric diffusion,
Pacific J. Math. 237 (2008), no. 2, 327--347.

\bibitem {Stillinger et al 1}Stillinger F. H., Weber T. A., Hidden structure
in liquids, Phys. Rev. A 25, 978--989 (1982).

\bibitem {Stillinger et al 2}Stillinger F. H., Weber T. A., Packing structures
and transitions in liquids and solids, Science 225, 983--989 (1984).

\bibitem {Taibleson}Taibleson M. H., Fourier analysis on local fields.
Princeton University Press, 1975.

\bibitem {Taira}Taira Kazuaki, Boundary value problems and Markov processes.
Second edition. Lecture Notes in Mathematics, 1499. Springer-Verlag, 2009.

\bibitem {T-Z}Torba S. M., Z\'{u}\~{n}iga-Galindo W. A., Parabolic type
equations and Markov stochastic processes on adeles, J. Fourier Anal. Appl. 19
(2013), no. 4, 792--835.

\bibitem {Va1}Varadarajan V. S., Path integrals for a class of $p$-adic
Schr\"{o}dinger equations, Lett. Math. Phys. 39 (1997), no. 2, 97--106.

\bibitem {V-V-Z}Vladimirov V. S., Volovich I. V., Zelenov E. I., $p$-adic
analysis and mathematical physics. World Scientific, 1994.

\bibitem {Wales}Wales D. J., Miller M. A., Walsh T. R., Archetypal Energy
Landscapes, Nature 394, 758--760 (1998).

\bibitem {Yoshino}Yoshino H., Hierarchical Diffusion, Aging and
Multifractality, J. Phys. A 30, 1143--1160 (1997).

\bibitem {Zu}Z\'{u}\~{n}iga-Galindo W. A., Parabolic equations and Markov
processes over p-adic fields, Potential Anal. 28 (2008), no. 2, 185--200.

\bibitem {Z1}Z\'{u}\~{n}iga-Galindo W. A., The non-Archimedean stochastic heat
equation driven by Gaussian noise, J. Fourier Anal. Appl. 21 (2015), no. 3, 600--627.

\bibitem {Zuniga-LNM-2016}W. A. Z\'{u}\~{n}iga-Galindo, Pseudodifferential
equations over non-Archimedean spaces. Lectures Notes in Mathematics 2174,
Springer, 2016.
\end{thebibliography}
\end{document}